\documentclass[onefignum,onetabnum]{siamonline250106}

\usepackage{lipsum}
\usepackage{amsfonts}
\usepackage{graphicx}
\usepackage{epstopdf}
\usepackage{algorithmic}
\usepackage{amssymb,bm,bbm}
\usepackage{tikz-cd,tikz}
\usepackage{cite}

\ifpdf
  \DeclareGraphicsExtensions{.eps,.pdf,.png,.jpg}
\else
  \DeclareGraphicsExtensions{.eps}
\fi

\usepackage{enumitem}
\setlist[enumerate]{leftmargin=.5in}
\setlist[itemize]{leftmargin=.5in}

\crefname{subsection}{section}{sections}

\newsiamremark{remark}{Remark}
\newsiamremark{hypothesis}{Hypothesis}
\crefname{hypothesis}{Hypothesis}{Hypotheses}
\newsiamthm{claim}{Claim}

\newsiamthm{thm}{Theorem}
\newsiamthm{cor}{Corollary}
\newsiamremark{prop}{Proposition}
\newsiamremark{rmk}{Remark}
\newsiamremark{conj}{Conjecture}

\headers{BCMs with Topic-Weighted Discordance}{G. J. Li, J. Luo, W. Chu}

\title{Bounded-Confidence Models of Multidimensional Opinions with Topic-Weighted Discordance
\thanks{Submitted to the editors on January 31, 2025.
\funding{The first and second authors acknowledge funding from National Science Foundation (NSF) grant 1829071. The second author acknowledges funding from NSF grant number 2136090 and NSF grant number 1922952 through the Algorithms for Threat Detection (ATD) program. A portion of this work was performed under the auspices of the U.S. Department of Energy by Lawrence Livermore National Laboratory under Contract DE-AC52-07NA27344 (LLNL-JRNL-2002439).}}
}

\author{
Grace Jingying Li\thanks{Lawrence Livermore National Laboratory, Livermore, CA, USA.}
\and Jiajie Luo\thanks{Knowledge Lab, University of Chicago, Chicago, IL, USA.}
\and Weiqi Chu\thanks{Department of Mathematics and Statistics, University of Massachusetts Amherst, Amherst, MA, USA.}
}

\usepackage{amsopn}

\newcommand{\N}{\mathcal{N}}
\newcommand{\R}{\mathbb{R}}

\newcommand{\bu}{\bm{u}}
\newcommand{\bx}{\bm{x}}
\newcommand{\bX}{\bm{X}}
\newcommand{\by}{\bm{y}}

\newcommand{\bR}{\mathbb{R}}

\renewcommand{\vec}[1]{\bm{#1}}

\def\r{1.5}
\def\x{4}

\ifpdf
\hypersetup{
  pdftitle={Bounded-Confidence Models of Multidimensional Opinions with Topic-Weighted Discordance},
  pdfauthor={G. J. Li, J. Luo, and W. Chu}
}
\fi

\begin{document}

\maketitle

\begin{abstract}
People's opinions on a wide range of topics often evolve over time through their interactions with others. Models of opinion dynamics primarily focus on one-dimensional opinions, which represent opinions on one topic. However, opinions on various topics are rarely isolated; instead, they can be interdependent and correlated. In a bounded-confidence model (BCM) of opinion dynamics, agents are receptive to each other only if their opinions are sufficiently similar. We extend classical agent-based BCMs --- namely, the Hegselmann--Krause BCM, which has synchronous interactions, and the Deffuant--Weisbuch BCM, which has asynchronous interactions --- to a multidimensional setting, in which the opinions are multidimensional vectors representing opinions of different topics and opinions on different topics are interdependent. To measure opinion differences between agents, we introduce topic-weighted discordance functions that account for opinion differences in all topics. We define regions of receptiveness for our models, and we use them to characterize the steady-state opinion clusters and provide an analytical approach to compute these regions. In addition, we numerically simulate our models on various networks with initial opinions drawn from a variety of distributions. When initial opinions are correlated across different topics, our topic-weighted BCMs yield significantly different results in both transient and steady states compared to baseline models, where the dynamics of each opinion topic are independent.
\end{abstract}

\begin{keywords}
bounded-confidence models, opinion dynamics, multidimensional opinions, non-Euclidean opinion distance, social networks
\end{keywords}

\begin{MSCcodes}
91D30, 05C82, 37H05
\end{MSCcodes}

\section{Introduction}\label{sec:intro}

People exchange and update their opinions on various topics through social interactions. The study of opinion dynamics examines how opinions and viewpoints evolve and spread among individuals based on such interactions \cite{noorazar2020_review}. 
In an agent-based model (ABM) of opinion dynamics, each agent (which may represent an individual, a social media account, or other social entity) holds time-dependent opinions that change based on interactions with other agents. 
One can restrict interactions between agents by incorporating network structures in ABMs.
In such models, the nodes of a network represent agents and only adjacent nodes (i.e., nodes with an edge between them) can interact and, in turn, influence each other's opinions.
The framework of ABMs provides a quantitative method for analyzing opinion evolution as a dynamical system on networks \cite{noorazar2020_review,porter2016} and offers valuable insights into decision-making processes \cite{urena2019review}, opinion formation \cite{lorenz2007continuous}, and the dissemination of ideas \cite{jia2015opinion,friedkin1990social}.

People tend to be receptive to and influenced by opinions similar to their own 
opinions \cite{selective_exposure_def}.
This psychological phenomenon is incorporated into bounded-confidence models (BCMs) of opinion dynamics \cite{noorazar2020_review, HK_model, deffuant2000mixing}. 
In agent-based BCMs, agents have continuous-valued opinions and are receptive to
each other if and only if their opinion differences are less than a threshold called a ``confidence bound''.
Two of the most well-known BCMs are the Hegselmann--Krause (HK) model \cite{krause2000, HK_model} and the Deffuant--Weisbuch (DW) model \cite{deffuant2000mixing}, both of which update in discrete time.
The HK model updates synchronously; at each time step, all agents update their opinions based on the opinions of their neighbors in a network.
By contrast, the DW model updates asynchronously; at each time step, only one pair of agents (represented as adjacent nodes in an underlying network) interacts and potentially updates their opinions based on their opinion differences.

BCMs have predominantly been studied with one-dimensional (1D) opinions. 
However, researchers have observed that opinions on multiple related topics can be interdependent and influence one another. 
Political scientist Philip Converse \cite{converse1964} used the term “belief system” to describe attitudes interconnected by constraints or interdependencies. 
Empirical studies have shown correlations between political positions \cite{benoit2006_book, benoit2012}, indicating that opinions on different topics may not be independent.
Researchers have extended BCMs to multidimensional opinions for both the DW model \cite{lorenz2006_multi, li2017_ja_vector} and the HK model \cite{fortunato2005_vector, lorenz2006_multi, bhattacharyya2013, etesami2013_hk_convergence, hegselmann2019_hk_d_dimension_convergence, brooks2020}, measuring the distance between two opinion vectors using the Euclidean distance.\footnote{For two $K$-dimensional opinion vectors $\bm{x}$ and $\bm{y}$, their Euclidean distance is defined by $\|\bm{x} - \bm{y}\|_2 = \sqrt{\sum_{i = 1}^K (x^i - y^i)^2}$, where $x^i$ and $y^i$ are the $i$th components of $\bm{x}$ and $\bm{y}$, respectively.} 
While researchers often make this choice for mathematical convenience and tractability, there is little empirical evidence supporting its appropriateness for modeling 
distances between opinions \cite{benoit2006_book}.

Researchers have also studied BCMs and other opinion models with multidimensional opinions using non-Euclidean distance metrics. 
For instance, some researchers \cite{fortunato2005_vector,lorenz_fostering_2008} studied multidimensional HK and DW models by measuring the opinion distances using the $L^1$ or $L^\infty$ norms.
Schweighofer et al. \cite{schweighofer2020} introduced multidimensional BCMs based on angular distance. 
In their models, agents' opinion vectors rotate towards each other when they compromise. 
Non-Euclidean distances have also been applied to a multidimensional variant \cite{parsegov2015} of the Friedkin--Johnson (FJ) model \cite{friedkin1990social}, which is a linear opinion dynamics model where agents consistently incorporate their initial opinions into each update.
Gubanov et al. \cite{gubanov_multidimensional_2021} proposed a model in which agents update their opinions as a weighted average of their initial opinions, their current opinions, and their neighbors' current opinions. The weights are determined by non-Euclidean opinion distances, with smaller distances resulting in larger weights.

Some researchers have incorporated topic correlations or coupling parameters into multidimensional opinion models to represent interdependencies between topics. 
Parsegov et al. \cite{parsegov2015} extended the FJ model by introducing a correlation matrix $C$, where each entry $C_{pq}$ quantifies the influence of topic $q$ on topic $p$. 
Ye et al. \cite{ye2020} later studied a continuous-time version of the FJ model inspired by the model developed by Parsegov et al.
Baumann et al. \cite{baumann2021} explored a model featuring a topic-overlapping matrix $\Phi$, where each entry $\Phi_{pq}$ encodes the angle between topics $p$ and $q$ in a latent space, which measures the effect of topic $p$ on topic $q$.
Chen et al. \cite{chen2021_multi_agent} formulated
a model featuring time-dependent topic correlations.
In their model, agents update their opinions based on these correlations, external interventions, and a modified Jager--Amblard rule \cite{jager2005}, which applies either an attraction or repulsion mechanism depending on the Euclidean distance between opinions.
These studies have largely used topic correlations in opinion models to determine how much opinions change and the extent to which agents influence each other. 
By contrast, we incorporate topic interdependencies into the opinion distance functions in our BCMs that determine whether or not agents are receptive to each other; topic interdependencies do not affect how much the opinions of receptive agents change.
We are able to use topic interdependencies to affect receptiveness because BCMs have a confidence bound for determining receptiveness, which yields different dynamics than the aforementioned opinion models with topic correlation.
To our knowledge, there has been no prior work incorporating topic dependency into a continuous opinion distance function to determine receptiveness in BCMs.

In our BCMs, when two agents interact on a topic $k$, their opinion differences on other topics affect their receptiveness to each other on topic $k$. Our models seek to incorporate the sociological idea that individuals that are similar to each other are more likely to influence each others' opinions \cite{Axelrod1997}.
A prominent, pioneering model that incorporates this notion on multiple topics is the Axelrod model \cite{Axelrod1997}, in which an agent adopts another agent's opinion on a given topic with probability equal to the proportion of topics on which they agree.
Additionally, we are aware of one BCM that uses topic interdependency to determine receptiveness.
Tan and Cheong \cite{Tan2018} formulated and studied a multidimensional BCM based on the HK model in which agents' receptiveness on one topic depends on the opinion differences on all topics. 
However, the way in which our BCMs determine receptiveness fundamentally differs from that of Tan and Cheong (see \cref{footnote:tan_and_cheong_model}).

The standard HK and DW models were initially simulated using uniformly random initial opinions \cite{HK_model, deffuant2000mixing}, and researchers simulating 1D BCMs still commonly use uniformly random initial opinions \cite{lorenz2007continuous,shang2013_dw_critical,bernardo2024}.
However, several studies of BCMs have examined non-uniform distributions of initial opinions.
Such approaches have been used to extend both the DW model \cite{jacobmeier2006, shang2013_dw_critical, carro2013, sobkowicz2015, hickok2022bounded, chu2023density} and the HK model \cite{kou2012_hk, yang2014_hk_consensus, blondel2010}. 
The distribution of initial opinions plays a critical role in shaping opinion evolution, influencing both transient behaviors and steady states.
For example, several researchers have provided sufficient conditions on an initial opinion distribution that ensure the models converge to a consensus --- 
in which every agent converges to the same opinion --- in various contexts, including a mean-field BCM~\cite{gomez-serrano2012}, a BCM on hypergraphs adapted from the DW model~\cite{hickok2022bounded}, and both discrete- and continuous-time BCMs adapted from the HK model~\cite{blondel2010, yang2014_hk_consensus}.

In this paper, we propose multidimensional BCMs --- that are based on and extend the standard HK and DW models on networks --- that incorporate non-Euclidean distances and topic dependencies to measure opinion dissimilarity.
We introduce topic-weighted discordance functions $d_k$ for each topic $k$, which emphasize the opinion difference for topic $k$ while also considering the opinion differences in other topics. 
In our BCMs, agents compromise their opinions on a topic $k$ if their discordance, given by $d_k$, is less than the confidence bound.
Our topic-weighted discordance functions highlight the idea that more (respectively, less) receptivity in one topic can lead to more (respectively, less) receptivity in another topic.
To examine the effects of topic correlation on our BCMs, in our numerical simulations, we consider initial opinions in which the opinion values of different topics are correlated or independently sampled. 
Our numerical results (see \cref{sec:sim_details} and \cref{sec:numerical_results}) demonstrate that the choice of initial opinion distribution also significantly influences the evolution and formation of opinions in our multidimensional models.

This paper proceeds as follows. In \cref{sec:modeldescription}, we review the standard HK and DW models, with synchronous and asynchronous update rules, respectively, and introduce our multidimensional BCMs with topic-weighted discordance functions. In \cref{sec:theoretical}, we analyze the mathematical properties of our topic-weighted BCMs, including their convergence guarantees, formation of opinion clusters, and regions of receptiveness.
We provide the implementation details of our numerical simulations in \cref{sec:sim_details} and present our numerical results in \cref{sec:numerical_results}.
Finally, we provide concluding remarks in \cref{sec:conclusions}.
Our code and plots are available at \url{https://gitlab.com/graceli1/topic-weighted-BCMs/}.

\section{Model descriptions}
\label{sec:modeldescription}
In this section, we review the standard Hegselman--Krause (HK) \cite{HK_model} and Deffuant--Weisbuch (DW) \cite{deffuant2000mixing} models and introduce our BCMs with topic-weighted opinion discordance. 
We study our BCMs on networks in which each node represents an agent and each edge represents a social connection between the incident nodes.
Let $G = (V,E)$, where $V$ is the set of nodes and $E$ is the set of edges, denote a time-independent, unweighted, and undirected graph (i.e., network) without self-edges or multi-edges. We denote $M=|V|$ as the total number of nodes.

\subsection{Standard BCMs with a single opinion topic}\label{sec:standard_models}

In the standard HK and DW models, each node $i$ has a time-dependent opinion $x_i(t)$ with scalar values in the closed interval $[0,1]$. 
Both models have a \emph{confidence bound} $c$ that controls the ``open-mindedness'' of nodes to different opinions.
We say that nodes $i$ and $j$ are 
\emph{receptive} to each other at time $t$ if their opinion difference is less than the confidence bound $c$ (i.e., if $|x_i(t) - x_j(t)| < c$).

\subsubsection{The standard Hegselmann--Krause (HK) model}\label{sec:1D_HK}

The standard HK model \cite{krause2000, HK_model} is a discrete-time synchronous BCM on 
a time-independent, unweighted, and undirected graph $G = (V,E)$
with no self-edges or multi-edges.\footnote{The HK model was examined initially on a fully-mixed population \cite{HK_model}, but we use its extension to networks (see e.g., \cite{fortunato2005, parasnis_hegselmann-krause_2018, schawe2021_HK_networks})
as our ``standard HK model''.}
The set of
nodes to which node $i$ is receptive at time $t$ is\footnote{In \cite{krause2000, HK_model}, 
$\N_i(t) = \{ i \} \cup \left\{j: \ |x_i(t) - x_j(t)| \leq c { \text{ and }(i,j)\in E} \right\}$. 
We use a strict inequality to be consistent with the strict inequality in the DW model.}
\begin{equation}\label{eq: HK_receptive_set_baseline}
  \N_i(t) = \{i\} \cup \left\{j : (i,j)\in E 
  \text{~and~} |x_i(t)-x_j(t)| < c\right\}\,.
\end{equation}
In other words, $\N_i(t)$ consists of node $i$ itself along with all adjacent nodes to which $i$ is receptive.
At each time $t$, we update the opinion of each node $i \in \{1,2, \ldots, M\}$ by calculating
\begin{equation}\label{eq:HK_baseline_update_rule}
	x_{i}(t + 1) = |\N_i(t)|^{-1} \sum_{j \in \N_i(t)} x_{j}(t)\,.
\end{equation}

\subsubsection{The standard Deffuant--Weisbuch (DW) model}\label{sec:1D_DW}

The standard DW model \cite{deffuant2000mixing} is a discrete-time asynchronous BCM on a time-independent, 
unweighted, and undirected graph $G = (V,E)$ with no self-edges or multi-edges. At each time $t$, we choose an edge $(i,j) \in E$ uniformly at random and
update the opinions of nodes $i$ and $j$
by calculating
\begin{equation} \label{eq: DW_baseline_update_rule}
  \begin{aligned}
       x_i(t+1) &= \begin{cases}
          x_i(t) + \mu (x_j(t) - x_i(t)) \,, 
            & \textrm{if~} |x_i(t)-x_j(t)| < c  \\
          x_i(t) \,, & \textrm{otherwise} \,,
      \end{cases} \\
      x_j(t+1) &= \begin{cases}
          x_j(t) + \mu (x_i(t) - x_j(t)) \,,
            & \textrm{if~} |x_i(t)-x_j(t)| < c  \\
          x_j(t) \,, & \textrm{otherwise} \,,
      \end{cases} 
  \end{aligned}
  \end{equation}
where $\mu \in (0, 0.5]$ is the compromise parameter\footnote{Alternatively, one can consider $\mu \in (0,1)$ as in Meng et al. \cite{meng2018}, although this is an uncommon choice. When $\mu > 0.5$, nodes ``overcompromise'' when they change their opinions; they overshoot the mean opinion and change which side of the mean opinion they are on.}.
The opinions of all other nodes remain unchanged.
The compromise parameter $\mu$ controls how much nodes adjust their opinions when they interact with a node to whom they are receptive.
When $\mu = 0.5$, two interacting nodes that are receptive to each other precisely average their opinions; when $\mu \in (0,0.5)$, interacting nodes that are receptive to each other move towards each others' opinions, but they do not adopt the mean opinion.
Unlike in the HK model, the asynchronous update rule \eqref{eq: DW_baseline_update_rule} of the DW model updates only the opinions of a pair of nodes at each time step.

\subsection{BCMs with topic-weighted opinion discordance}\label{sec:topic_weighted_models}
Having described the standard HK and DW models in \cref{sec:standard_models}, we now introduce our variants of these models with topic-weighted opinion discordance.
Our topic-weighted BCMs are defined on a time-independent, unweighted, and undirected graph $G = (V,E)$ with no self-edges or multi-edges. The nodes of a graph represent agents with opinions on $K$ related topics. Let $\bx_i(t) \in \R^K$ denote the opinion vector of node $i$ at time $t$; its $k$th entry, which we denote by $x_i^k(t)$, is the opinion of node $i$ on topic $k$. 
Unlike the standard BCMs (see \cref{sec:standard_models}), which have opinions in the interval $[0,1]$, we formulate our models to consider opinions in $\R^K$.

Our models have a \emph{topic weight} $\omega\in [0,1]$, which is a parameter that determines how much weight is placed on the opinion difference in a topic $k$ to determine receptivity on that topic.
For a fixed topic weight $\omega$, we define topic-weighted distance functions $d_k: \R^K \times \R^K \mapsto \R$ that we call \emph{discordance functions}.\footnote{Our choice to call these functions ``discordance'' functions was inspired by Hickok et al. \cite{hickok2022bounded}.}
For each topic $k$, the discordance function $d_k$ takes two opinion vectors, $\bx_i$ and $\bx_j$, as inputs and calculates a distance between them that emphasizes topic $k$ while accounting for the other topics. 
The discordance function $d_k$ is defined as
\begin{equation}\label{eq: d_k}
d_k(\bx_i, \bx_j) = 
\omega |x_i^k - x_j^k| + \frac{1-\omega}{K} \sum_{\ell=1}^K |x_i^\ell - x_j^\ell|\,.
\end{equation}
In our models, when nodes $i$ and $j$ have opinions that satisfy $d_k(\bx_i, \bx_j) < c$, we say that nodes $i$ and $j$ are \emph{receptive} to each other on topic $k$.\footnote{
\label{footnote:tan_and_cheong_model}
Tan and Cheong \cite{Tan2018} also formulated a multidimensional BCM with interdependent topics, but defined the receptiveness of nodes in a fundamentally different approach from ours.
In their model, each topic $k$ has an associated confidence bound $c_k$ and an inclusivity parameter $\alpha_k$. Nodes $i$ and $j$ are receptive on topic $k$ if 
$|x_i^k(t)-x_j^k(t)|<c_k$ and 
$|x_i^\ell(t)-x_j^\ell(t)|<\alpha_k c_\ell$
for each topic $\ell\neq k$.
Thus, two agents are receptive on topic $k$ if their opinion differences on each topic is less than some threshold for that topic. By contrast, our models compute the discordance function as a weighted sum of opinion differences of all topics \eqref{eq: d_k} and check that this discordance is within a single confidence-bound threshold.
\nopagebreak
}
The \emph{maximum discordance function} $d_\mathrm{max}$ is
\begin{equation}\label{eq: d_max}
  d_\mathrm{max}(\bx_i, \bx_j) = \max\limits_{1 \leq k \leq K} d_k(\bx_i, \bx_j) \,.
\end{equation}
When nodes $i$ and $j$ have opinions that satisfy $d_\mathrm{max}(\bx_i, \bx_j) < c$, nodes $i$ and $j$ are receptive to each other on all topics.

\subsubsection{Our HK model with topic-weighted opinion discordance}\label{sec:topic_HK_model}

We now describe our HK model with topic-weighted opinion discordance. We refer to this model as our \emph{topic-weighted HK model}. Our topic-weighted HK model has three parameters: the confidence bound $c \in [0,\infty)$, the compromise parameter $\mu \in (0, 1]$, and the topic weight $\omega \in [0,1]$.

At each discrete time, we update the opinions of all nodes with weighted averages of their opinions and the opinions of their receptive neighbors. 
Specifically, for each node $i$, we define $\N_i^k(t)$ as the set that contains node $i$ itself and all adjacent nodes to which node $i$ is receptive on topic $k$ at time $t$. That is,
\begin{equation}\label{eq: HK_receptive_set}
  \N_i^k(t) = \{i\} \cup \left\{j : (i,j)\in E 
  \text{~and~} d_k(\bx_i(t),\bx_j(t)) < c\right\} \,.
\end{equation}
We update the opinions of all nodes $i$ and all topics $k$ with the update rule
\begin{equation}\label{eq: HK_update_rule}
  x_i^k(t+1) = (1-\mu)x_i^k(t)+\frac{\mu}{|\N_i^k(t)|}
  \sum_{j \in \N_i^k(t)} x_j^k(t) \,.
\end{equation} 
Our update rule \eqref{eq: HK_update_rule} is similar to that of Chazelle and Wang \cite{chazelle2017_inertial_hk}, but we use our topic-weighted discordance functions instead of Euclidean distance to determine the nodes in a receptive set \eqref{eq: HK_receptive_set}.
The compromise parameter $\mu$ determines how much nodes adjust their opinions towards the mean opinion of the neighbors to whom they are receptive. 
The standard HK model (see \cref{sec:1D_HK}) does not have this compromise parameter.
When the compromise parameter $\mu=1$, the update rule \eqref{eq: HK_update_rule} of each node for every topic in the topic-weighted HK model reduces to the update rule \eqref{eq:HK_baseline_update_rule} of the standard HK model. 
When the topic weight $\omega = 1$, the discordance \eqref{eq: d_k} reduces to $d_k(\bx_i, \bx_j) = |x_i^k - x_j^k|$, so the opinion discordance for topic $k$ has no dependence on other topics. 
Therefore, our topic-weighted HK model with $\omega = 1$ and $\mu = 1$ is equivalent to modeling the opinion dynamics of each topic using the standard HK model (see \cref{sec:1D_HK}).
We refer to our topic-weighted HK model with $\omega = 1$ and $\mu = 1$ as our \emph{baseline HK model}.

\subsubsection{Our DW model with topic-weighted opinion discordance}\label{sec:topic_dw_model}

We now describe our DW model with topic-weighted opinion discordance. We refer to this model as our \emph{topic-weighted DW model}. Our topic-weighted DW model has three parameters: 
the confidence bound $c \in [0, \infty)$, the compromise parameter $\mu \in (0,0.5]$, and the topic weight $\omega \in [0,1]$.
The confidence bound and compromise parameter in our topic-weighted DW model are analogous to those BCM parameters in the standard DW model (see \cref{sec:1D_DW}).

At each discrete time, we select one pair of adjacent nodes and update their opinion values on a single topic.
We select an edge $(i,j) \in E$ uniformly at random and a topic $k \in \{1, \dots, K\}$ uniformly at random. 
We then update the opinions of nodes $i$ and $j$ on topic $k$ by calculating
\begin{equation} \label{eq: DW_update_rule}
  \begin{aligned}
       x_i^k(t+1) &= \begin{cases}
          x_i^k(t) + \mu \left[x_j^k(t) - x_i^k(t)\right] \,, 
            & \textrm{if~} d_k(x_i(t),x_j(t)) < c \\
          x_i^k(t) \,, & \textrm{otherwise} \,,
      \end{cases} \\
      x_j^k(t+1) &= \begin{cases}
          x_j^k(t) + \mu \left[x_i^k(t) - x_j^k(t)\right] \,,
            & \textrm{if~} d_k(x_i(t),x_j(t)) < c  \\
          x_j^k(t) \,, & \textrm{otherwise} \,.
      \end{cases} 
  \end{aligned}
  \end{equation}
We refer to our topic-weighted DW model with $\omega = 1$ as our \emph{baseline DW model}. 
When $\omega = 1$, the discordance function $d_k(\bx_i, \bx_j) = |x_i^k - x_j^k|$ depends solely on the opinions for topic $k$ and not on other topics. 
In our baseline DW model, when we focus exclusively on the time steps when topic $k$ is selected, the restricted dynamics follow the standard DW model.
However, the dynamics of opinions of different topics are not independent. In particular, if topic $k$ has an opinion update at time $t$, then topic $\ell\neq k$ cannot be updated at $t$, since only one topic is chosen for interaction at each time.

\section{Theoretical properties of our models}
\label{sec:theoretical}

We now discuss theoretical properties of our BCMs.
In \cref{sec:theoretical:limit_opinions}, we discuss the asymptotic behavior of the opinions in our BCMs, including the convergence of opinion values. 
In \cref{sec:theoretical:effective_graphs}, we discuss the structural properties of the ``effective graphs'' (the time-dependent subgraphs that capture mutual receptivity between nodes) in our BCMs, which help motivate the framework for our numerical simulations.
In \cref{sec:theoretical:region_of_receptiveness}, we discuss ``regions of receptiveness'', which characterize the set of opinion vectors to which a node is receptive at a time $t$.

\subsection{Convergence analysis and absorbing states}
\label{sec:theoretical:limit_opinions}

We define the \emph{opinion state} in our topic-weighted BCMs to be a concatenated vector with the opinions of all nodes on all topics. 
We denote the opinion state as 
$\bX(t) = (\bx_1(t),\bx_2(t),\ldots,\bx_M(t))^T \in\bR^{MK}$.
We use the following result from Lorenz \cite{lorenz2005stabilization} to prove that the opinion states in our topic-weighted BCMs converge in time.

\begin{thm}[Lorenz \cite{lorenz2005stabilization}]\label{thm: lorenzthm}
Let $\{ A(t) \}_{t=0}^\infty$ be a sequence of row-stochastic matrices with $A(t)\in \mathbb{R}_{\geq 0}^{M \times M}$.
Consider the following conditions on the sequence $\{ A(t) \}_{t=0}^\infty$.
\begin{enumerate}
\item[(1)] The diagonal entries of $A(t)$ are positive.
\item[(2)] For each $i, j \in \{1,\ldots,M\}$, we have that $[A(t)]_{ij} > 0$ if and only if $[A(t)]_{ji} > 0$.
\item[(3)] There is a constant $\delta > 0$ such that the smallest positive entry of $A(t)$ for each $t$ is larger than $\delta$. 
\end{enumerate} 
Given times $t_0$ and $t_1$ with $t_0 < t_1$, let 
\begin{equation}
	A(t_0, t_1) = A(t_1-1)\times A(t_1-2)\times\cdots\times A(t_0)\,.
\end{equation}
If conditions (1)--(3) are satisfied, then there exists a time $t'$ and pairwise-disjoint classes $\mathcal{I}_{1} \cup \cdots \cup \mathcal{I}_{p} = \{1,\ldots,M\}$ of indices such that if we reindex the rows and columns of the matrices in the order $\mathcal{I}_{1}, \ldots, \mathcal{I}_{p}$, then
\begin{equation}
	\lim _{t \rightarrow \infty} A(0, t) = \left[\begin{array}{ccc}
		K_{1} & & 0 \\
		& \ddots & \\
		0 & & K_{p}
	\end{array}\right] A\left(0, t'\right)\,,
\end{equation}
where each $K_q$, with $q \in \{1, 2, \dots , p\}$, is a row-stochastic matrix of size $|\mathcal{I}_{q}|\times |\mathcal{I}_{q}|$ whose rows are all the same.
\end{thm}

Using \cref{thm: lorenzthm}, we now prove that for our topic-weighted HK and DW models, the 
limit $\bX^*:=\lim_{t\rightarrow\infty} \bX(t)$ of the opinion state exists.

\begin{thm}\label{thm: HK_existence_steady}
    Let $\bX(t)$ be the opinion state of our topic-weighted HK model with update rule \eqref{eq: HK_update_rule} and initial opinion state $\bX(0)$. It follows that the limit $\bX^* = \lim_{t\rightarrow\infty} \bX(t)$ exists.
\end{thm} 

\begin{proof}

    Let $i$ and $j$ denote node indices and $k$ denote a topic index.
    We can express the update rule \eqref{eq: HK_update_rule} in the matrix form 
    $\bX(t+1) = A(\bX(t),t) \bX(t)$, where 
    \begin{equation}
    \left[A(\bX(t),t)\right]_{\alpha\beta} = \begin{cases}
        1 - \mu + \mu/|\N_i^k(t)| \,,
        &\text{if~} \alpha = \beta=(i-1)K+k  \\
        \mu/|\N_i^k(t)| \,,
        &\text{if~} \alpha = (i-1)K+k\,, ~~ \beta = (j-1)K+k\,, \\
        & 
        \text{and~}
        j\in \N_i^k(t) \setminus \{i\}  \\
        0 \,,
        & \text{otherwise}\,.
    \end{cases}
    \end{equation}
    It is readily checked that each matrix $A(\bX(t),t)$ satisfies the conditions in \cref{thm: lorenzthm}.
    Therefore, the opinion state $\bX(t)$ converges to a limit state $\bX^*$ as time $t$ goes to infinity.
\end{proof}

\begin{thm}\label{thm: DW_existence_steady}
Let $\bX(t)$ be the opinion state of our topic-weighted DW model with update rule \eqref{eq: DW_update_rule} and initial opinion state $\bX(0)$. It follows that the limit $\bX^*=\lim_{t\rightarrow\infty} \bX(t)$ exists.
\end{thm}

\begin{proof}
We can express the update rule \eqref{eq: DW_update_rule} in the matrix form 
$\bX(t+1) = A(\bX(t),t) \bX(t)$. 
Suppose that we select edge $(i,j)$ and topic $k$ at discrete time $t$. If the discordance function $d_k(\bx_i(t),\bx_j(t))$ is greater than or equal to $c$, then $A(\bX(t),t)=I_{MK}$ is the identity matrix of dimension $MK\times MK$. 
If the discordance function $d_k(\bx_i(t),\bx_j(t))$ is less than $c$, then 
\begin{equation}
    \left[A(\bX(t),t)\right]_{\alpha\beta} = \begin{cases}
       1-\mu \,,
         & \text{if~} \alpha=\beta=(i-1)K+k \\
       1-\mu \,,
         & \text{if~} \alpha=\beta=(j-1)K+k \\
       \mu \,,
         & \text{if~} \alpha=(i-1)K+k, ~~\beta=(j-1)K+k \\
       \mu \,,
         & \text{if~} \alpha=(j-1)K+k, ~~\beta=(i-1)K+k  \\
       \delta_{\alpha\beta} \,, & \text{otherwise}\,,
    \end{cases}
\end{equation}
where $\delta_{\alpha\beta}$ is the Kronecker delta.
It is readily checked that each matrix $A(\bX(t),t)$ satisfies the conditions in \cref{thm: lorenzthm}. Therefore, the opinion state $\bX(t)$ converges to a limit state $\bX^*$ as time $t$ goes to infinity.
\end{proof}

Theorem \ref{thm: HK_existence_steady} and Theorem \ref{thm: DW_existence_steady} give us the existence of a \emph{limit opinion} $\bx_i^* = \lim_{t\to\infty} \bx_i(t)$ for each node $i$ in our topic-weighted HK model and topic-weighted DW model, respectively.
For each distinct limit opinion, we say that a maximal set of nodes that have the same limit opinion is a \emph{limit opinion cluster}.
Suppose that there are $r$ distinct limit opinion clusters, denoted by $S_1, \ldots, S_r \subset V$.
We have that $\bigcup_{a=1}^r S_a = V$ and $S_a \bigcap S_b = \emptyset$ if $a \neq b$.
For a limit opinion cluster $S_a$, we denote its corresponding limit opinion by $\bm{z}_a \in \mathbb{R}^K$. For each node $i \in S_a$ and for each topic $k$, we have $x_i^{*k} = z_a^{k}$.

We say that an opinion state $\bX\in \bR^{MK}$
in our BCMs is an \emph{absorbing state} if for all edges $(i,j)\in E$ and all topics $k\in\{1,\ldots,K\}$, either $x_i^k=x_j^k$ or $d_k(\bx_i, \bx_j) \ge c$. If $\bX(T)$ is an absorbing state, then $\bX(t) = \bX(T)$ for all $t \ge T$. 
We now prove that the limit state $\bX^*$ is an absorbing state in our topic-weighted HK model and is almost surely an absorbing state in our topic-weighted DW model.

\begin{lemma}\label{lemma:limit_opinion_difference}
    In our topic-weighted HK model with update rule \eqref{eq: HK_update_rule},
    there is a time $T\geq 0$ such that for all edges $(i,j)\in E$ and for all topics $k$, if $x_i^{*k}\neq x_j^{*k}$, we have
    \begin{equation} \label{eq: finitet}
   d_k(\bx_i(t),\bx_j(t)) \ge c \text{~~for all~~} t\ge T\,.
    \end{equation}
\begin{proof}
    From Theorem \ref{thm: HK_existence_steady}, we know the limit $\bX^*=\lim_{t\rightarrow\infty} \bX(t)$ exists. Suppose there is exactly one distinct limit opinion cluster with corresponding limit opinion $\bx^*$. 
    Then for all nodes $i$ and topics $k$, we have $\lim_{t\rightarrow\infty} x_i^k = x^{*k}$. Therefore, the inequality \eqref{eq: finitet} is vacuously true with any $T$.

    Suppose that there are $r$ ($r\ge 2$) distinct limit opinion clusters, denoted by $S_1, \ldots, S_r \subseteq V$. Let $\bm{z}_a$ be the limit opinion of nodes that belong to $S_a$ (i.e.,  $\bm{z}_a=\lim_{t\rightarrow\infty}\bx_i(t)$ for all $i\in S_a$). 
    We define
    \begin{equation} \label{eq: delta_min}
      \delta = \min\limits_{\substack{1 \leq k \leq K, \\ 1 \leq a, b \leq r}} \{|{z}_{a}^{k} - {z}_{b}^{k}| : {z}_{a}^{k} \neq {z}_{b}^{k} \}
    \end{equation} 
    as the minimum distance between distinct limit opinions among all topics.
    Because there are at least two distinct limit opinion clusters, we know that $\delta > 0$. Since $\bX^*=\lim_{t\rightarrow \infty}\bX(t)$, for any $\varepsilon>0$, there exists
    a time $T \geq 0$ such that
  \begin{equation} \label{eq: epsilon}
      \begin{aligned}
          |x^{*k}_i-x^{k}_i(t)|<\varepsilon
      \end{aligned}
  \end{equation}
  for all nodes $i$, all topics $k$, and all times $t \geq T$.
  In particular, we choose $\varepsilon$ in \eqref{eq: epsilon} as $\varepsilon = \mu\delta/[2M(\mu+1)]$, where $M$ is the number of nodes.

Suppose that at some time $t\ge T$, there exists a topic $k$ and a node $i$ that has a neighbor $j$ (i.e., $(i,j)\in E$) with $x_i^{*k}\neq x_j^{*k}$ and
 \begin{equation}\label{eq: lemmaassup} 
  d_k(\bx_i(t),\bx_j(t)) < c \,.
  \end{equation} 
  We fix this time $t$ and topic $k$ and define a set of nodes
  \begin{equation}
  \Omega = \{i: i \text{~has a neighbor~} j \text{~with~} x_i^{*k}\neq x_j^{*k} \text{~and~} d_k(\bx_i(t),\bx_j(t)) < c\}\,.
  \end{equation}
  With the assumption above, we know $\Omega$ is not empty and define node $x_{i_{\min}}$ as
  \begin{equation}
  x_{i_{\min}}^{*k} = \min_{i\in\Omega} x_i^{*k}\,.
  \end{equation}
  Without loss of generality, we let $i_{\min}=1$ and $i_{\min}\in S_1$. 
  Node 1 has at least one neighbor $j$ with $x_1^{*k}\neq x_j^{*k}$ and $j \notin S_1$.
  Thus, $\N_1^k(t){\setminus}S_1$ has at least one element $j$ and therefore, is not empty. 
  From the update rule \eqref{eq: HK_update_rule}, we have
  \begin{equation} \label{eq: xii}
  \begin{aligned}
  \left|x_1^k(t+1)-x_1^k(t)\right| 
  &= \frac{\mu}{|\N_1^k(t)|}
  \left| \sum_{j\in\N_1^k(t){\setminus}S_1} \left[x_j^k(t)-x_1^k(t)\right]  + \sum_{j\in\N_1^k(t)\cap S_1} \left[x_j^k(t)-x_1^k(t)\right] \right| \\
  &\ge \frac{\mu}{|\N_1^k(t)|}
  \left| \sum_{j\in\N_1^k(t){\setminus}S_1} \left[x_j^k(t)-x_1^k(t)\right] \right| - \frac{\mu}{|\N_1^k(t)|} \sum_{j\in\N_1^k(t)\cap S_1} \left|x_j^k(t)-x_1^k(t)\right|\,.
  \end{aligned}
  \end{equation}
Using \eqref{eq: epsilon} and then \eqref{eq: delta_min}, for 
    $j\in\N_1^k(t){\setminus}S_1$ we have 
    \begin{equation} \label{eq: part1}
    x_j^k(t)-x_1^k(t) > x_j^{*k}-x_1^{*k}-2\varepsilon > \delta - 2\varepsilon >0\,.
    \end{equation}
    For $j\in\N_1^k(t){\cap}S_1$, since both $x_j^k(t)$ and $x_1^k(t)$ are in the same limit opinion cluster, we have $x_i^{*k} = x_j^{*k}$. Consequently, using \eqref{eq: epsilon}, we have 
    \begin{equation}\label{eq: part2}
    \left|x_j^k(t)-x_1^k(t)\right| < 2\varepsilon \,.
    \end{equation}
    Combining \eqref{eq: xii}, \eqref{eq: part1}, and \eqref{eq: part2}, we obtain
    \begin{equation}  \label{eq: x11}
    \left|x_1^k(t+1)-x_1^k(t)\right| > \frac{\mu\delta}{M} - 2\mu\varepsilon\,.
    \end{equation}

By our choice of $\varepsilon = \mu\delta/[2M(\mu+1)]$, we simplify \eqref{eq: x11} and obtain
\begin{equation}
    \left|x_1^k(t+1)-x_1^k(t)\right| > \frac{\mu\delta}{M(\mu+1)} \,.
    \label{eq: xiii}
\end{equation}
We simultaneously have 
  \begin{equation}
  \left|x_1^k(t+1)-x_1^k(t)\right| \le \left|x_1^{*k}-x_1^k(t+1)\right| + \left|x_1^{*k}-x_1^k(t)\right| < 2\varepsilon = \frac{\mu\delta}{M(\mu+1)}\,,
  \end{equation}
  which contradicts \eqref{eq: xiii}. Therefore, the assumption in \eqref{eq: lemmaassup} does not hold. 
\end{proof}
\end{lemma}

\begin{thm} \label{thm: HK-absorbing}
In our topic-weighted HK model with update rule \eqref{eq: HK_update_rule} and initial opinion state $\bX(0)$, the limit state $\bX^* = \lim_{t\rightarrow\infty} \bX(t)$ is an absorbing state.
\end{thm}
\begin{proof}
By \cref{lemma:limit_opinion_difference}, there exists a time $T\geq 0$ such that for all edges $(i,j)\in E$ and all topics $k$, we have either (1) $x_i^{*k} = x_j^{*k}$; or (2) $x_i^{*k} \neq x_j^{*k}$ and $d_k(\bx_i(t),\bx_j(t)) \ge c$ for all $t\ge T$. 
It therefore follows that if $x_i^{*k}\neq x_j^{*k}$, then
\begin{equation}
    d_k(\bx^*_i,\bx^*_j) = \lim_{t\to\infty} d_k(\bx_i(t),\bx_j(t)) \geq c\,.
\end{equation}
Therefore, $\bX^*$ is an absorbing state.
\end{proof}

\begin{thm} \label{thm: DW-absorbing}
In our topic-weighted DW model with update rule \eqref{eq: DW_update_rule} and initial opinion state $\bX(0)$, the limit state $\bX^* = \lim_{t\rightarrow\infty} \bX(t)$ is almost surely an absorbing state.
\end{thm}
We illustrate the main idea of the proof of \cref{thm: DW-absorbing} in \cref{fig: proof_diag}.
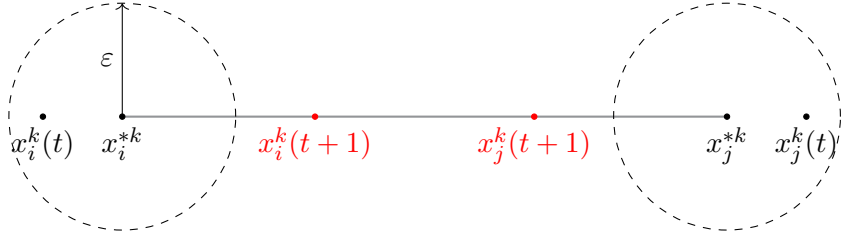
\begin{figure}[H]
    \centering
    \begin{tikzpicture}
    \draw[gray, thick] (-\x,0) -- (\x,0); 
    
    \draw[black,dashed] (-\x,0) circle (\r); 
    \draw[black,dashed] (\x,0) circle (\r); 
    
    \filldraw [black] (-\x,0) circle (1pt) node[anchor=north]{$x_i^{*k}$}; 
    \filldraw [black] (\x,0) circle (1pt) node[anchor=north]{$x_j^{*k}$};
    
    \filldraw [black] (-\x-\r*0.7,0) circle (1pt) node[anchor=north]{$x_i^{k}(t)$}; 
    \filldraw [black] (\x+\r*0.7,0) circle (1pt) node[anchor=north]{$x_j^{k}(t)$};
    
    \filldraw [red] (-\x+\r*1.7,0) circle (1pt) node[anchor=north]{$x_i^{k}(t+1)$}; 
    \filldraw [red] (\x-\r*1.7,0) circle (1pt) node[anchor=north]{$x_j^{k}(t+1)$};

    \draw[->] (-\x,0) -- (-\x,\r); 
    \node at (-\x*1.05,\r/2) {$\varepsilon$};
    \end{tikzpicture}
    \caption{This figure shows a scenario when the limit $\bX^*$ of 
    an opinion state is not an absorbing state. 
    Because $\bX^*$ is not an absorbing state, there is some edge $(i,j) \in E$ and some topic $k$ such that $x_i^{*k}\neq x_j^{*k}$ and $d_k(\bx_i^*,\bx_j^*)<c$.
    For large enough $t$, and for each node $i$ and each topic $k$, the opinion $x_i^k(t)$ will always remain close to its limit $x_i^{*k}$.
    The two dashed circles mark the intervals centered at $x_i^{*k}$ and $x_j^{*k}$ with a radius of $\varepsilon$. We seek to choose $\varepsilon$ small enough such that (1) nodes $i$ and $j$ are receptive on topic $k$ (i.e., $d_k(\bx_i(t),\bx_j(t))<c$), so that an interaction between them will lead to their opinions changing;
    and (2) the new opinions $x_i^k(t+1)$ and $x_j^k(t+1)$ (in red) fall outside of the two circles.
    This would contradict the opinion $x_i^k(t)$ remaining close to its limit opinion $x_i^{*k}$ (with difference no more than $\varepsilon$), implying that the assumed scenario (i.e., $\bX^*$ is not an absorbing state) does not occur.
    }
    \label{fig: proof_diag}
\end{figure}

\begin{proof} 
  From Theorem \ref{thm: DW_existence_steady}, we know that the limit state $\bX^*$ exists.
  Suppose that $\bX^*$ is not an absorbing state. Then there must exist an edge $(i,j)\in E$ and a topic $k$, such that
  \begin{equation}
  |x^{*k}_i - x^{*k}_j| = \delta >0 \text{~~and~~} d_k(\bx^{*}_i, \bx^{*}_j)=c_0<c\,.
  \end{equation}
  Since $\bX^* = \lim_{t\rightarrow \infty} \bX(t)$, for any $\varepsilon>0$, there exists a time $T \geq 0$, such that
  \begin{equation}\label{eq:dw_absorbing_epsilon}
      \begin{aligned}
          |x^{*k}_i-x^{k}_i(t)|<\varepsilon\,, \quad |x^{*k}_j-x^{k}_j(t)|<\varepsilon\,, \quad 
          d_k(\bx_i^*,\bx_i(t))<\varepsilon\,, \quad d_k(\bx_j^*,\bx_j(t))<\varepsilon
      \end{aligned}
  \end{equation}
  for all times $t \ge T$. In particular, we choose 
  \begin{equation}\varepsilon=\min\left\{\frac{c-c_0}{3}\,, ~\frac{\mu\delta}{2(1-\mu)}\right\}\,.\end{equation}
   By the Borel--Cantelli lemma, with probability 1, in our topic-weighted DW model, we choose every edge in $E$ and every topic infinitely often. 
   Therefore, we almost surely choose the edge $(i,j)$ and the topic $k$ at some time $t\geq T$. 
   At this time $t$, we have that $i$ and $j$ interact on topic $k$, yielding
   \begin{equation}
   \begin{aligned}
   d_k(\bx_i(t),\bx_j(t)) &\le d_k(\bx_i^*,\bx_i(t))+d_k(\bx_j^*,\bx_j(t))+d_k(\bx_i^*,\bx_j^*) \\
   & \le \frac{c-c_0}{3} + \frac{c-c_0}{3} + c_0 < c\,,
   \end{aligned}
   \end{equation}
   so we update the opinions $x_i^k(t)$ and $x_j^k(t)$ according to the update rule \eqref{eq: DW_update_rule}. 
   From equations \eqref{eq: DW_update_rule} and \eqref{eq:dw_absorbing_epsilon}, we have
   \begin{equation}\label{eq: less}
       \begin{aligned}
           |x_i^k(t+1)-x_j^k(t+1)| &= (1-2\mu)|x_i^k(t)-x_j^k(t)| \\
           & \le (1-2\mu)\left(|x_i^{*k}-x_j^{*k}|+|x_i^{*k}-x_i^k(t)|+|x_j^{*k}-x_j^k(t)|\right)\\
           & < (1-2\mu)\left(\delta+2\varepsilon\right) \le \frac{\delta(1-2\mu)}{1-\mu}\,.
       \end{aligned}
   \end{equation}
   From \eqref{eq:dw_absorbing_epsilon}, we also have
   \begin{equation}\label{eq: greater}
       \begin{aligned}
           |x_i^k(t+1)-x_j^k(t+1)| &\ge |x_i^{*k}-x_j^{*k}|-|x_i^{*k}-x_i^k(t+1)|-|x_j^{*k}-x_j^k(t+1)| \\
           & > \delta - 2\varepsilon \ge  \frac{\delta(1-2\mu)}{1-\mu}\,.
       \end{aligned}
   \end{equation}
   Equations \eqref{eq: less} and \eqref{eq: greater} cannot hold simultaneously, and therefore, there can be no interaction between nodes $i$ and $j$ on topic $k$ at any time $t\geq T$. 
   Because $i$ and $j$ interact almost surely on topic $k$ at some time $t\geq T$, we have that $\bX^*$ is almost surely an absorbing state.
\end{proof}

\subsection{Effective graphs}\label{sec:theoretical:effective_graphs}
Let $G=(V,E)$ be a graph in our topic-weighted BCMs. 
In our topic-weighted BCMs, an opinion state $\bX$ has an associated \emph{effective graph} $G_\mathrm{eff}$, which is a subgraph of $G$ that only has edges between adjacent nodes (with respect to $G$) that are receptive to each other on all topics.\footnote{Other researchers have referred to an effective graph for BCMs with 1D opinions as a ``confidence graph'' \cite{bernardo2024}, a ``communication graph'' \cite{bhattacharyya2013}, and a ``corresponding graph'' \cite{yang2014_hk_consensus}.} 
That is, an edge $(i,j) \in E$ is in the effective graph $G_\mathrm{eff}$ if and only if $d_{\max}(\bx_i,\bx_j)< c$.
For a given graph $G=(V,E)$, we define a mapping $\phi: \bR^{MK}\rightarrow \{0,1\}^{M\times M}$ as
\begin{equation}
\left[\phi(\bX)\right]_{ij} = \begin{cases}
        1 & \text{if~} (i,j)\in E \text{~and~} d_{\max}(\bx_i,\bx_j) < c \\
        0 & \text{otherwise}\,.
    \end{cases}
\end{equation}

Let $G_\mathrm{eff}(t)$ and $G_\mathrm{eff}^*$ be the effective graphs associated with 
a time-dependent opinion state $\bX(t)$ and its
limit state $\bX^* = \lim_{t\rightarrow\infty} \bX(t)$, respectively.
The effective graph $G_\mathrm{eff}(t)$ is a subgraph of $G$ with the adjacency matrix $A_{\mathrm{eff}}(t)$, where
\begin{equation} \label{eq: effective-adjacency}
A_{\mathrm{eff}}(t)=\phi(\bX(t))\,.
\end{equation}
We call a set of nodes that make up a connected component of the effective graph $G_\mathrm{eff}(t)$ an \emph{opinion cluster} at time $t$.
The effective graph $G_\mathrm{eff}^*$ associated with a limit state $\bX^*$ is a subgraph of $G$ with the adjacency matrix $A_{\mathrm{eff}}^*$, where
\begin{equation} \label{eq: effective-adjacency2}
A_{\mathrm{eff}}^*=\phi(\bX^*)\,.
\end{equation} 
From \cref{thm: HK-absorbing}, a limit state $\bX^*$ in our topic-weighted HK model is an absorbing state. 
Similarly, from \cref{thm: DW-absorbing},
a limit state $\bX^*$ in our topic-weighted DW model is almost surely an absorbing state. 
If $\bX^*$ is an absorbing state, then for each edge $(i,j)$ in $G$, we have either (1) $d_\mathrm{max}(\bx_i^*, \bx_j^*) = 0$, and therefore, the edge $(i,j)$ is in $G_\mathrm{eff}^*$; or (2) $d_\mathrm{max}(\bx_i^*, \bx_j^*) \geq c$, and therefore the edge $(i,j)$ is not in $G_\mathrm{eff}^*$. 
Recall that for each distinct limit opinion, the set of nodes that have that limit opinion is a limit opinion cluster. Therefore, the nodes that make up a connected component of $G_\mathrm{eff}^*$ belong to the same limit opinion cluster. 
However, there may not be a one-to-one correspondence between limit opinion clusters and connected components of $G_\mathrm{eff}^*$.
In particular, if two or more connected components of $G_\mathrm{eff}^*$ have the same corresponding limit opinion, then they all correspond to the same limit opinion cluster.

\begin{thm}\label{thm: HK-eff-graph}
    Let $\bX(t)$ be the opinion state of our topic-weighted HK model with update rule \eqref{eq: HK_update_rule} and initial opinion state $\bX(0)$
    and define $A_{\mathrm{eff}}(t)$ and $A_{\mathrm{eff}}^*$ as in \eqref{eq: effective-adjacency} and \eqref{eq: effective-adjacency2}. Then it holds that
    \begin{equation} 
    \lim_{t\rightarrow\infty} A_{\mathrm{eff}}(t) 
    = A_{\mathrm{eff}}^* \,.
    \end{equation}
\end{thm}

\begin{proof}
    By Theorem \ref{thm: HK-absorbing}, we have that for all edges $(i,j)\in E$, 
    the maximum discordance of the corresponding limit opinions satisfies either (1) $d_{\max}(\bx_i^*,\bx_j^*) = 0$; or (2) $d_{\max}(\bx_i^*,\bx_j^*) \geq c$. We consider these two cases separately.

    When $d_{\max}(\bx_i^*,\bx_j^*)=0$, we have $\left[A_{\mathrm{eff}}^*\right]_{ij}=1$ and
    \begin{equation}
        \lim_{t\rightarrow\infty} d_{\max}(\bx_i(t),\bx_j(t)) = d_{\max}(\bx_i^*,\bx_j^*) = 0\,,
    \end{equation}
    which implies that there exists a time $T \geq 0$ such that $d_{\max}(\bx_i(t),\bx_j(t)) < c$ for all $t \geq T$. Therefore, $\left[A_{\mathrm{eff}}(t)\right]_{ij}=1$ for all $t \geq T$ and $\lim_{t\rightarrow\infty}\left[A_{\mathrm{eff}}(t)\right]_{ij}=1$.

    When $d_{\max}(\bx_i^*,\bx_j^*)=c_0\geq c$, we have $\left[A_{\mathrm{eff}}^*\right]_{ij}=0$. 
    Moreover, there must be a topic $k$ such that $\bx_i^{*k}\neq \bx_j^{*k}$. 
    By \cref{lemma:limit_opinion_difference}, there exists a time $T$ such that for all $t\geq T$, we have $d_k(\bx_i(t),\bx_j(t))\geq c$. 
    It therefore follows that $\left[A_{\mathrm{eff}}(t)\right]_{ij}=0$ for all $t\geq T$, so $\lim\limits_{t\to\infty}\left[A_{\mathrm{eff}}(t)\right]_{ij}=0$. 
\end{proof}

\begin{thm} \label{thm: DW-eff-graph}
    Let $\bX(t)$ be the opinion state of our topic-weighted DW model with update rule \eqref{eq: DW_update_rule} and initial opinion state $\bX(0)$, and define $A_{\mathrm{eff}}(t)$ and $A_{\mathrm{eff}}^*$ as in \eqref{eq: effective-adjacency} and \eqref{eq: effective-adjacency2}. 
    Then it almost surely holds that
    \begin{equation} 
    \lim_{t\rightarrow\infty}A_{\mathrm{eff}}(t) 
    = A_{\mathrm{eff}}^*\,.
    \end{equation}
\end{thm}
\begin{proof}
    From Theorem \ref{thm: DW-absorbing}, we have for all edges $(i,j)\in E$, the maximum discordance $d_{\max}(\bx_i^*,\bx_j^*)$ is either $0$ or greater than or equal to $c$. We prove this theorem by considering three cases for edge $(i,j)\in E$: (1) $d_{\max}(\bx_i^*,\bx_j^*)=0$; (2) $d_{\max}(\bx_i^*,\bx_j^*)>c$; and (3) $d_{\max}(\bx_i^*,\bx_j^*)=c$.

    When $d_{\max}(\bx_i^*,\bx_j^*)=0$, we have $\left[A_{\mathrm{eff}}^*\right]_{ij}=1$ and
    \begin{equation}
        \lim_{t\rightarrow\infty} d_{\max}(\bx_i(t),\bx_j(t)) = d_{\max}(\bx_i^*,\bx_j^*) = 0\,,
    \end{equation}
    which implies that there exists time $T \geq 0$, such that $d_{\max}(\bx_i(t),\bx_j(t)) < c$ for all $t \geq T$. Therefore, $\left[A_{\mathrm{eff}}(t)\right]_{ij}=1$ for all $t \geq T$ and 
    $\lim_{t\rightarrow\infty}\left[A_{\mathrm{eff}}(t)\right]_{ij}=1$.

    When $d_{\max}(\bx_i^*,\bx_j^*)=c_0>c$, we have 
    $\left[A_\mathrm{eff}^*\right]_{ij} = 0$.
    Furthermore, we have
    \begin{equation}
        \lim_{t\rightarrow\infty} d_{\max}(\bx_i(t),\bx_j(t)) = d_{\max}(\bx_i^*,\bx_j^*) = c_0\,.
    \end{equation}
    There exists a time $T \geq 0$, such that $d_{\max}(\bx_i(t),\bx_j(t)) > c + (c_0-c)/2$ for all $t \geq T$. 
    Therefore, 
    $\lim_{t\rightarrow \infty}\left[A_\mathrm{eff}(t)\right]_{ij}
    = 0 = \left[A_\mathrm{eff}^*\right]_{ij}$.
    
    When $d_{\max}(\bx_i^*,\bx_j^*)=c$, we have $\left[A_\mathrm{eff}^*\right]_{ij}=0$. We examine the limit superior of $\left[A_\mathrm{eff}(t)\right]_{ij}$ and will show that
    \begin{equation} \label{eq: limsup}
    \limsup_{t\rightarrow\infty} \left[A_\mathrm{eff}(t)\right]_{ij} = 0 \quad \text{almost surely}\,.
    \end{equation}
    As $\left[A_\mathrm{eff}(t)\right]_{ij}$ only takes values of 0 or 1, the limit superior is equal to 0 or 1. 
    Suppose that the limit superior is equal to $1$. 
    Then there exists an increasing sequence $\{t_s\}_{s=1}^\infty$ of times that goes to infinity such that 
    \begin{equation}
    \left[A_\mathrm{eff}(t_s)\right]_{ij} = 1 \text{~for~all~} s \ge 1\,.
    \end{equation}
    Since we have an infinite number of times in the sequence $\{t_s\}_{s=1}^\infty$ and a finite number of topics, there exists a topic $k$ such that 
    \begin{equation}\label{eq: dkts}
        d_k(\bx_i(t_{s_r}),\bx_j(t_{s_r}))<c  \text{~for all times~} t_{s_r}\,,
    \end{equation} 
    where $\{t_{s_r}\}_{r=1}^\infty \subset \{t_s\}_{s=1}^\infty$ is a subsequence (note that $\lim_{r\rightarrow\infty}t_{s_r}=\infty$).
    Because $i$ and $j$ interact on topic $k$ with a positive probability, we choose the pair $(i,j)$ to interact on the topic $k$ infinitely often. 
    We use a similar idea as in Figure \ref{fig: proof_diag} to show that this scenario almost never occurs.
    For $\varepsilon=c\mu/2(1-\mu)$, there exists a time $T \geq 0$     such that $|x_i^k(t)-x_i^{*k}|<\varepsilon$ for all 
    nodes $i$, all topics $k$, and all times $t \geq T$.
    With probability $1$, we choose the pair $(i,j)$ and the topic $k$ at some time $t \geq T$ in the subsequence $\{t_{s_r}\}_{r=1}^\infty$.
    When this occurs, by \eqref{eq: dkts}, we update the opinions using \eqref{eq: DW_update_rule} and obtain
    \begin{equation}
        |x_i^k(t+1)-x_j^k(t+1)| = |x_i^k(t+1)-x_j^k(t+1)|(1-2\mu) < (c+2\varepsilon)(1-2\mu) = \frac{c(1-2\mu)}{1-\mu}\,.
    \end{equation}
    However, this contradicts
    \begin{equation}
        |x_i^k(t+1)-x_j^k(t+1)| \ge |x_i^{*k}-x_j^{*k}| - |x_i^{*k}-x_i^k(t+1)| - |x_j^{*k}-x_j^k(t+1)| > c-2\varepsilon = \frac{c(1-2\mu)}{1-\mu}\,.
    \end{equation}
    Therefore, the limit superior is equal to 0 almost surely, which proves \eqref{eq: limsup}. As a result, 
    when $d_{\max}(\bx_i^*,\bx_j^*)=c$,
    we almost surely have $\lim_{t\rightarrow\infty} \left[A_\mathrm{eff}(t)\right]_{ij}=0=\left[A_\mathrm{eff}^*\right]_{ij}$. 
    This completes the proof.
    \end{proof}

\subsection{Region of receptiveness}\label{sec:theoretical:region_of_receptiveness}

We now define and examine the region of receptiveness for opinion vectors, which is an important concept to characterize the limit opinion vectors.
Fixing an opinion vector $\bm{x}$, we define the \emph{region of receptiveness} for $\bm{x}$ as
\begin{equation}\label{eq: region_of_absorption}
    \mathcal{R}_{\bm{x}} = \{\bm{y} \in \mathbb{R}^{K}: ~d_\mathrm{max}(\bm{x},\bm{y}) < c \} \,,
\end{equation}
which, for a node $i$ with opinion $\bm{x}$, is the region of opinions that a node $j$ can have so that $i$ and $j$ are receptive to each other on all topics. 
We omit the dependence of the parameters of the BCMs (namely, $c$ and $\omega$) in our definition to simplify the notation.

We define $|\bm{x}-\bm{y}|$ as the vector with $k$th entry $|\bm{x}-\bm{y}|^k=|x^k-y^k|$. 
In order to characterize the region of receptiveness, we first define the set
\begin{equation}
\mathcal{U} = \left\{|\bm{x} - \bm{y}|: \bx,\by\in\mathbb{R}^K \text{~and~} d_{\max}(\bx,\by)<c \right\}\,.
\end{equation}
Notice that using the notation of $\mathcal{U}$, we have $\mathcal{R}_{\bx} = \{\bm{y}\in\mathbb{R}^K: |\bm{x}-\bm{y}| \in \mathcal{U}\}$. 
\begin{lemma}\label{lemma:u_x}
    Consider our topic-weighted BCMs (with update rules \eqref{eq: HK_update_rule} and \eqref{eq: DW_update_rule}). 
    Let $\bm{x}\in\R^K$ be a fixed opinion vector. 
    If $\omega\neq 0$, then the closure of $\mathcal{U}$ is a bounded convex polytope, and the set of the vertices of the closure of $\mathcal{U}$ is
    $\bigcup_{s=0}^K \mathcal{V}_s$, where 
    \begin{equation}\label{eq:vertex_feasible_thm}
        \mathcal{V}_s = \left\{\bm{u} \in \mathbb{R}^K : 
        s \textup{~entries of~} \bm{u} \textup{~are~}
        \frac{cK}{s+(K-s)\omega}
        \textup{~and the remaining~} K-s \textup{~entries are~} 0 \right\} \,.
    \end{equation}
    If $\omega=0$, then the set of vertices of $\mathcal{U}$ is
    $\mathcal{V}_0\cup\mathcal{V}_1$.
\end{lemma}

\begin{proof}
The set $\mathcal{U}_x$ is the feasible region characterized by the constraints
\begin{align}
    \label{eq:d_k_constraint_strict}
       \omega u^k + \frac{1-\omega}{K} \sum_{\ell=1}^K u^\ell &< c\,,  \quad \text{for~} k=1,\ldots,K \\
    \label{eq:equality_constraint}
       u^k &\geq 0\,,  \quad \text{for~} k=1,\ldots,K\,.
\end{align}
In addition, the closure of $\mathcal{U}$ is the feasible region with the constraints \eqref{eq:equality_constraint} and 
    \begin{equation}\label{eq:d_k_constraint}
       \omega u^k + \frac{1-\omega}{K} \sum_{\ell=1}^K u^\ell \leq c\,, \quad
       \text{for~} k=1,\ldots,K \,.
    \end{equation} 
    Therefore, $\mathcal{U}$ is a convex polytope.
    In addition, if we define the norm\footnote{\label{footnote:weird_norm}It is readily checked that $\|\bm{u}\|_{d,k}$ is a norm. Moreover, $\|\cdot\|_{d,k}$ induces the $k$th discordance function $d_k$ (i.e., $d_k(\bx,\by) = \|\bx-\by\|_{d,k}$).}
    \begin{equation}
    \|\bm{u}\|_{d,k} = \omega|u_k|+\frac{1-\omega}{K}\sum\limits_{k=1}^K|u_k| \,,
    \end{equation}
    we have $\|u\|_{d,k}\leq c$ for all $u\in \mathcal{U}$. 
    This implies that $\mathcal{U}$ and its closure are bounded.

    We now examine the vertices of the closure of $\mathcal{U}$. 
    Each vertex is a point in $\mathbb{R}^K$ and determined as a solution of a system of $K$ linear equations selected from the following $2K$ equations (obtained from \eqref{eq:equality_constraint} and \eqref{eq:d_k_constraint}):
    \begin{align}
    \label{eq:d_k_system}
       \omega u^k + \frac{1-\omega}{K} \sum_{\ell=1}^K u^\ell &= c\,, \quad \text{for~} k=1,\ldots,K \\
    \label{eq:equality_system}
        u^k &= 0\,, \quad \text{for~} k=1,\ldots,K\,.
    \end{align}
    The vertices are solutions of such linear systems that are
    feasible with respect to the constraints \eqref{eq:equality_constraint} and \eqref{eq:d_k_constraint}.

    We first show that the solution of a linear system that satisfies both equations \eqref{eq:d_k_system} and \eqref{eq:equality_system} for the same topic $k$ is not feasible with respect to the constraints \eqref{eq:equality_constraint} and \eqref{eq:d_k_constraint}. 
    For $\omega = 1$, if a linear system contains equations \eqref{eq:d_k_system} and \eqref{eq:equality_system} for the same topic $k$, then these two equations reduce to $u^k = c$ and $u^k = 0$, which has no solution.
    Now, consider $\omega\neq1$.
    Seeking a contradiction, suppose that both equations \eqref{eq:d_k_system} and \eqref{eq:equality_system} hold for 
    the same topic $k$.
    By combining \eqref{eq:d_k_system} and \eqref{eq:equality_system} for this topic $k$, we obtain
    \begin{equation}\label{eq:combined_equations}
        \frac{1-\omega}{K}\sum_{\ell=1}^K u^\ell = c\,.
    \end{equation}
    Because each $u^\ell$ is non-negative (because of the constraints given by \eqref{eq:equality_constraint}) and ${cK}/{(1-\omega)}$ is positive, we can choose a topic $p$ such that $u^p > 0$. 
    Because $u^k=0$ and equation \eqref{eq:combined_equations} holds, we observe that 
    \begin{equation}
        \omega u^p + \frac{1-\omega}{K} \sum_{\ell=1}^K u^\ell =\omega u^p + c\,,
    \end{equation}
    which is strictly greater than $c$ because $u^p>0$. 
    This contradicts the constraint \eqref{eq:d_k_constraint} for topic $p$ and implies that there is no feasible solution to a linear system containing both equations \eqref{eq:d_k_system} and \eqref{eq:equality_system} for the same topic $k$.

    We now show that when $\omega\neq0$, a linear system of $K$ equations consisting of precisely one of \eqref{eq:d_k_system} or \eqref{eq:equality_system} for each $k\in\{1,\ldots,K\}$ yields a feasible solution and the set of solutions of such systems is precisely the set $\bigcup_{s=0}^K\mathcal{V}_s$ (see \eqref{eq:vertex_feasible_thm}).
    Consider such a system with $s$ equations of the form \eqref{eq:d_k_system} and $K-s$ equations of the form $\eqref{eq:equality_system}$. 
    If $s=0$, we have $u^k=0$ for all $k$, which yields a feasible solution whose entries are all $0$. 
    Consider $s\geq1$.
    By reindexing variables, we consider, without loss of generality, a system with equations \eqref{eq:d_k_system} for $k=1,\ldots,s$
    and \eqref{eq:equality_system} for $k=s+1,\ldots,K$.
    Because we have $u^k=0$ for all $k\in\{s+1,\ldots,K\}$, we only need to solve for $u^1,\ldots,u^s$ such that 
    \begin{equation}\label{eq:reduced_system}
        \omega u^k + \frac{1-\omega}{K}\sum_{\ell=1}^s u^\ell = c\,, \quad \text{for~} k=1,\ldots,s\,.
    \end{equation}

    Let $\bm{u}_{\textrm{tr}}=(u^1,\ldots,u^s)^T$ denote a truncated vector of the first $s$ topics. 
    We can rewrite \eqref{eq:reduced_system} as the matrix equation $(\omega I + \bm{v}\bm{w}^T)\bm{u}_{\mathrm{tr}}=c\bm{1}$, where $\bm{v}=\frac{1-\omega}{K}\bm{1}$, $\bm{w}=\bm{1}$, and $\bm{1}$ is an $s$-dimensional vector whose entries are all $1$. 
    By the Sherman--Morrison formula\cite{bartlett1951,hagner1989}, the matrix $\omega I + \bm{v}\bm{w}^T$ is invertible with its inverse equal to  
    \begin{equation}\label{eq:inverse_Sherman}
    \left(\omega I + \bm{v}\bm{w}^T\right)^{-1} = \frac{1}{\omega} I - \frac{1-\omega}{s\omega+(K-s)\omega^2}  
    \bm{1} \bm{1}^T\,.
    \end{equation}
    By multiplying this matrix inverse to $c\bm{1}$, we obtain the entries of $\bm{u}_{\textrm{tr}}$ and have that
    \begin{equation} \label{eq: uk_entries}
    \begin{aligned}
    u^k &= \frac{cK}{s+(K-s)\omega}\,, && \text{for}~ k = 1,\ldots,s \\
    u^k &= 0\,, &&\text{for}~ k = s+1,\ldots,K\,.
    \end{aligned}
    \end{equation}
    Our solution is non-negative and therefore satisfies \eqref{eq:equality_constraint} for all $k$. 
    Using \eqref{eq: uk_entries}, we verify that 
    \begin{equation}
        \omega u^k + \frac{1-\omega}{K}\sum_{\ell=1}^K u^\ell = c\,, \quad \text{for~} k=1,\ldots,K\,,
    \end{equation}
    which implies that our solution satisfies \eqref{eq:d_k_constraint} for all $k$. 
    Therefore, our solution $\bm{u}$ is feasible, and the set of vertices of the closure of $\mathcal{U}$ is $\bigcup_{s=0}^K \mathcal{V}_s$.

    We now show that for $\omega=0$, the vertices of the closure of $\mathcal{U}$ are given by $\mathcal{V}_0\cup\mathcal{V}_1$. 
    We observe that when $\omega=0$, equation \eqref{eq:d_k_system} reduces to the same equation for all $k$. 
    Therefore, to have a system of $K$ distinct equations, we have at most $1$ equation of the form \eqref{eq:d_k_system} with the rest taking the form \eqref{eq:equality_system}. 
    If a system of equations has exactly one equation of the form \eqref{eq:d_k_system}, then the solution has $K-1$ zero entries and only one nonzero entry whose value is $cK$.
    If a system of equations has no equation of the form \eqref{eq:d_k_system}, then the entries of the solution are all $0$. 
    In both cases, the solutions are feasible.
    Therefore, the set of vertices of the closure of $\mathcal{U}$ is $\mathcal{V}_0\cup\mathcal{V}_1$.
\end{proof}

We now characterize the region of receptiveness of an opinion vector $\bm{x}$ with the following theorem.

\begin{thm}\label{thm:region_receptiveness}
For our topic-weighted BCMs (with update rules \eqref{eq: HK_update_rule} and \eqref{eq: DW_update_rule}), the region of receptiveness $\mathcal{R}_{\bm{x}}$ (see \eqref{eq: region_of_absorption}) for an opinion vector $\bm{x} \in \mathbb{R}^K$ is the interior of a bounded convex polytope. 
If $\omega\neq0$, the vertices of $\mathcal{R}_{\bm{x}}$ are precisely $\bm{y}$ such that $|\bm{x}-\bm{y}|\in\mathcal{V}_s$ 
(see \eqref{eq:vertex_feasible_thm}) 
for some $s\in\{1,\ldots,K\}$.
If $\omega=0$, the vertices of $\mathcal{R}_{\bm{x}}$ are precisely $\bm{y}$ such that $|\bm{x}-\bm{y}|\in\mathcal{V}_1$.
\end{thm}

\begin{proof}
    Fix an opinion vector $\bm{x}\in\R^K$. 
    For an opinion vector $\bm{y}\in\R^K$, let $\bm{u}=|\bm{x}-\bm{y}|$. 
    We have that $\bm{y}\in \mathcal{R}_{\bm{x}}$ if and only if $d_k(\bm{x},\bm{y})<c$ for all $k\in\{1,\ldots,K\}$. 

    Define $\tilde{\mathcal{U}}=\{\vec{\sigma}\odot \vec{u}: \vec{\sigma}\in\{1,-1\}^K, \vec{u}\in \mathcal{U}\}$, where $\odot$ denotes the entry-wise vector product and the entries of $\vec{\sigma}$ are either $1$ or $-1$. In other words, we obtain $\tilde{\vec{u}} \in \tilde{\mathcal{U}}$ by flipping the sign of one or multiple entries of $\vec{u}\in \mathcal{U}$.
    Using this notation, we have $\by\in \mathcal{R}_{\bx}$ if and only if there is some $\tilde{\bm{u}}\in \tilde{\mathcal{U}}$ such that 
    $\bm{y} = \bm{x} + \tilde{\bm{u}}$.
    Because $\mathcal{R}_{\bx}$ is a translation of $\tilde{\mathcal{U}}$, it is enough to show that $\tilde{\mathcal{U}}$ is the interior of a bounded convex polytope. 
    We note that $\tilde{\bu}\in\tilde{\mathcal{U}}$ if and only if $\tilde{\bu}$ satisfies
    \begin{equation}\label{eq:constraints_U_k_strict}
        \omega \sigma_k^S \tilde{u}^k + \frac{(1-\omega)}{K}\sum_{\ell=1}^K\sigma_\ell^S \tilde{u}^\ell < c
    \end{equation}
        for all topics $k = 1,\ldots,K$ and all subsets $S\subseteq\{1,\ldots,K\}$, where
    \begin{equation}
    \sigma_k^S = \begin{cases}
        -1 & \text{if~} k \in S\,, \\ 
        1 & \text{if~} k\notin S\,.
    \end{cases}
    \end{equation}
    
    We note that $\tilde{\mathcal{U}}$ is the interior of its closure $\overline{\tilde{\mathcal{U}}}$, 
    which is a convex polytope characterized by the linear constraints 
    \begin{equation}\label{eq:constraints_U_k}
        \omega \sigma_k^S \tilde{u}^k + \frac{(1-\omega)}{K}\sum_{\ell=1}^K\sigma_\ell^S \tilde{u}^\ell \leq c
    \end{equation}
    for all topics $k = 1,\ldots,K$ and all subsets $S\subseteq \{1,\ldots,K\}$.
    We also note that $\overline{\tilde{\mathcal{U}}}$ is bounded because the points it comprises satisfy $\|\bm{u}\|_{d,k}<c$ (see \cref{footnote:weird_norm}). 
    Therefore, $\tilde{\mathcal{U}}$ is the interior of a bounded convex polytope, and its vertices satisfy a linear system comprised of equations of the form \begin{equation}\label{eq:equation_sign_flipped}
        \omega \sigma_k^S \tilde{u}^k + \frac{(1-\omega)}{K}\sum_{\ell=1}^K\sigma_\ell^S \tilde{u}^\ell = c\,.
    \end{equation} 
    
    For a subset $S\subseteq \{1,\ldots,K\}$, we define the vector $\bm{1}_S$ as
    \begin{equation}
        [\bm{1}_S]^k = \begin{cases}
            -1 & k\in S\,, \\ 
            1 & \textrm{otherwise}\,.
        \end{cases}
    \end{equation} 
    We can express
    \begin{equation}
    \overline{\tilde{\mathcal{U}}} = \bigcup_{S\subseteq\{1,\ldots,K\}} \mathcal{U}^S\,,
    \end{equation}
    where 
    \begin{equation}
        \mathcal{U}^S=\bm{1}_S\odot\mathcal{U} = \{\bm{1}_S\odot\bm{u}: \bm{u}\in \mathcal{U}\}\,.
    \end{equation}
    Therefore, every vertex of $\overline{\tilde{\mathcal{U}}}$ is of the form $\bm{1}_S\odot \bm{v}$, where $S\subseteq \{1,\ldots,K\}$ and $\bm{v}$ is a vertex of $\mathcal{U}$. 

    Suppose $\bm{u}=\bm{1}_S\odot \bm{v}$, where $S\subseteq \{1,\ldots,K\}$ and $\bm{v}$ is a vertex of $\mathcal{U}$.
    By \cref{lemma:u_x}, we must have $\bm{v}\in \mathcal{V}_s$ for $s\in \{0,1,\ldots, K\}$.
    We claim that $\bm{u}$ is a vertex if and only if $s>0$. 
    Suppose $s>0$. 
    We note that $\bm{v}$ must be the solution of a linear system with $s$ equations of the form \eqref{eq:d_k_system} and $K-s$ equations of the form \eqref{eq:equality_system}.
    Without loss of generality, we can assume that $\bm{v}$ is a solution of the linear system consisting of \eqref{eq:d_k_system} for $k=1,\ldots,s$
    and \eqref{eq:equality_system} for $k=s+1,\ldots,K$.
    This implies that $\bm{u}$ is a solution to the system consisting of equations 
    \eqref{eq:equation_sign_flipped}
    for $k=1,\ldots,s$
    and \eqref{eq:equality_system} for $k=s+1,\ldots, K$.
    However, it is readily checked that this system is equivalent to the system consisting of equations 
    \eqref{eq:equation_sign_flipped}
    for $k=1,\ldots,s$ and equations
    \begin{equation}\label{eq:equation_sign_flipped_1}
        \omega \tau_1^k\sigma_k^S \tilde{u}^k + \frac{(1-\omega)}{K}\sum_{\ell=1}^K\tau_1^\ell\sigma_\ell^S \tilde{u}^\ell = c\,,
    \end{equation}
    for $k=s+1,\ldots, K$,
    where 
    \begin{equation}
        \tau_1^\ell = \begin{cases}
            -1 & \ell = 1 \,,
            \\ 1 & \text{otherwise}\,.
        \end{cases}
    \end{equation}
    (That is, \eqref{eq:equation_sign_flipped_1} is a special case of \eqref{eq:equation_sign_flipped} but with the sign on $u^1$ flipped.) 
    Therefore, $\bm{u}$ can be realized as the solution to a linear system that is comprised only of equations of the form \eqref{eq:equation_sign_flipped}, so it follows that $\bm{u}$ is a vertex of $\overline{\tilde{\mathcal{U}}}$. 
    Now, suppose $s=0$. Then $\bm{u}=0$, which does not satisfy any equation of the form \eqref{eq:equation_sign_flipped}. Therefore, $\bm{u}=0$ is not a vertex of $\overline{\tilde{\mathcal{U}}}$.
\end{proof}

We now provide an example of what the polytope $\mathcal{R}_{\bm{x}}$ looks like when $K = 2$. When $s = 2$, the entries of $\vec{u} \in \mathcal{V}_2$ (see \eqref{eq:vertex_feasible_thm}) are equal to $c$. 
This gives four vertices: $\bm{x} \pm (c, c)$ and $\bm{x} \pm (c, -c)$. 
When $s = 1$, the nonzero entries of 
$\vec{u} \in \mathcal{V}_1$ are equal to ${2c}/(\omega + 1)$. This gives four vertices at $\bm{x} \pm (0,{2c}/(\omega + 1))$ and $\bm{x} \pm ( 2c/(\omega + 1), 0)$. 
When $K = 2$, the region of receptiveness for $\bm{x}$ consists of a polygon with eight vertices. 
In this example, $\mathcal{R}_{\bm{x}}$ has 8 vertices and they are
\begin{equation*}
  \vec{x}\pm (c, c)\,, \quad  \vec{x}\pm (c, -c)\,, \quad
      \vec{x}\pm \left(0,\frac{2c}{\omega + 1}\right)\,, \quad \vec{x}\pm \left(\frac{2c}{\omega + 1}, 0\right)
\,.
\end{equation*}
In \cref{fig:region_of_absorption}, we show the region of receptiveness in 2D for a fixed value of $c$ and various values of $\omega$. When $\omega = 1$ (i.e., as in the baseline BCMs), the region of receptiveness is a square. As we decrease $\omega$, the area of the region of receptiveness increases; a node has a larger region of opinion space that a neighbor can be in for them to be mutually receptive on all topics.

\begin{figure}[ht]
  \centering
  \includegraphics[width=0.5\textwidth]{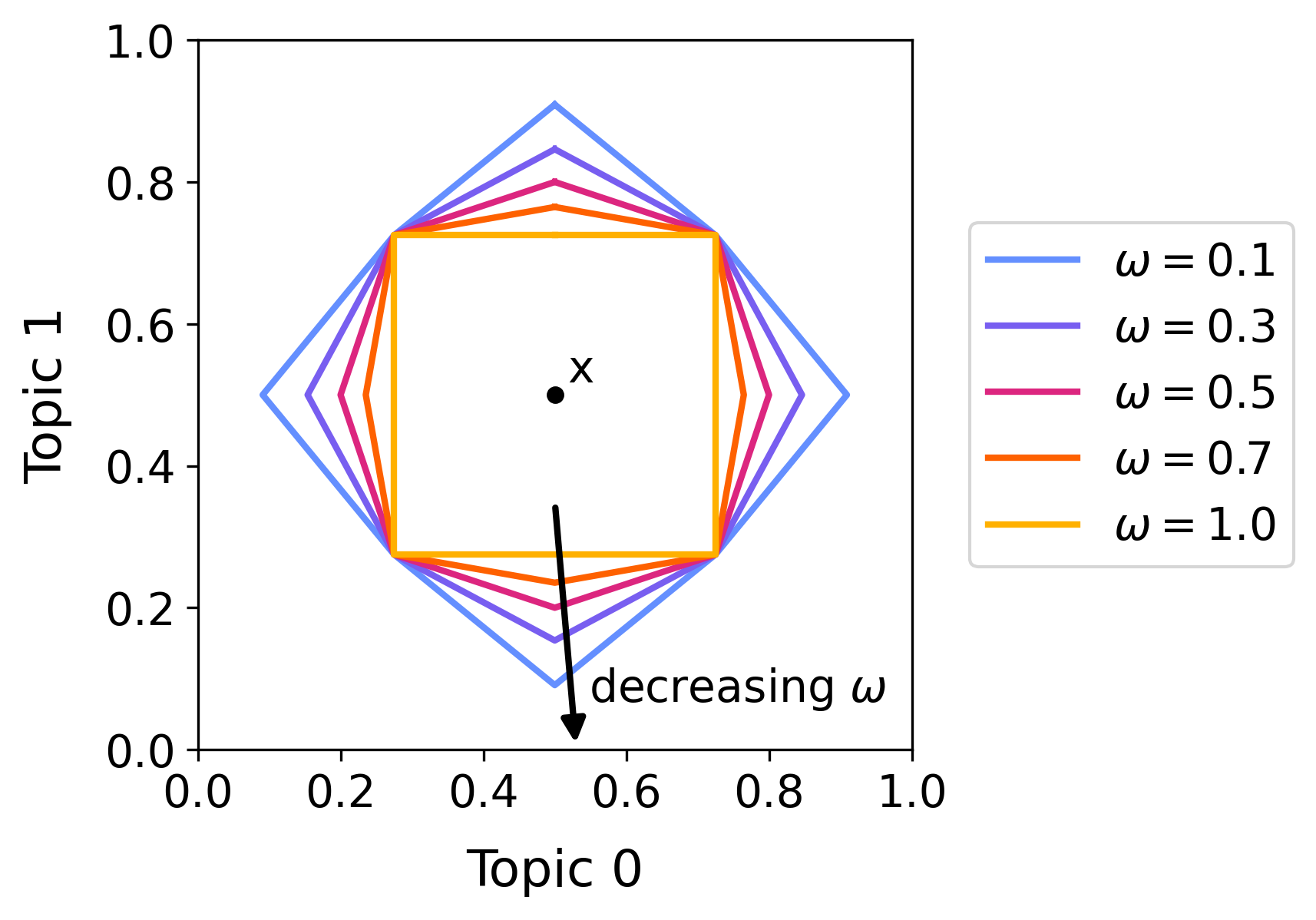}
  \caption[The regions of receptiveness in 2D for $c = 0.225$ and various values of $\omega$]{The regions of receptiveness for the opinion vector $\vec{x}=(0.5, 0.5)$ for $c = 0.225$ and various $\omega$ values. Each colored polygon represents the region of receptiveness for a different value of $\omega$. As we decrease $\omega$, the region of receptiveness grows.
  }
  \label{fig:region_of_absorption} 
\end{figure}

For the case $K = 3$, the region of receptiveness for the opinion vector $\bm{x}$ consists of a polygon with 26 vertices of the form $\bm{x} + \bm{u}$.
The $s = 3$ case yields 8 vertices, where each entry of $\bm{u}$ is $c$ or $-c$. 
The $s = 2$ case yields 12 vertices, where one entry of $\vec{u}$ is $0$ and the remaining two entries are ${3c}/({\omega + 2})$ or $-{3c}/({\omega + 2})$. 
For the case $s = 1$, it yields 6 vertices, where two entries of $\vec{u}$ are $0$ and the remaining entry is ${3c}/({2\omega + 1})$ or $-{3c}/({2\omega + 1})$. 

\section{Details of our numerical simulations}\label{sec:sim_details}

In this section, we discuss the setup for our numerical simulations of our topic-weighted BCMs. For our simulations, we consider opinions with $K = 2$ topics that lie in the region $[0,1] \times [0,1]$.
For this bounded opinion space, we correspondingly examine confidence bounds\footnote{
When $c = 0$, adjacent nodes $i$ and $j$ are never receptive to each other on any topic $k$.
When $c = 1$, adjacent nodes $i$ and $j$ are always receptive to each other on any topic $k$.
We do not examine these values of $c$.}
$c \in (0,1)$.

We simulate our topic-weighted BCMs with values of $\omega \in [0.1, 1]$. 
As we discussed in \cref{sec:topic_weighted_models}, when $\omega = 1$, the discordance function $d_k$ (see \eqref{eq: d_k}) depends only on the opinions on topic $k$. 
Our topic-weighted BCMs with $\omega = 1$ are the baseline BCMs.
When $\omega = 0$, all discordance functions $d_k$ are the same\footnote{For $\omega = 0$, when adjacent nodes $i$ and $j$ interact on any topic, they compromise their opinions on that topic if their mean opinion difference $\frac{1}{K} \sum_{\ell = 1}^K |x_i^\ell - x_j^\ell|$ is less than $c$.}; in our simulations, we do not consider this case.

In all our simulations of our topic-weighted HK model \eqref{eq: HK_update_rule}, we use the compromise parameter $\mu = 1$. 
At each time $t$, each node $i$ updates its opinion to be the average of the opinions of nodes in $\N_i^k(t)$ (see \eqref{eq: HK_receptive_set}).
In all our simulations of our topic-weighted DW model \eqref{eq: DW_update_rule}, we use the compromise parameter $\mu = 0.5$. 
When a pair of adjacent nodes are receptive to each other on a topic that they interact on, the nodes update their opinions on that topic to the average of their opinions on that topic.

\subsection{A stopping criterion for simulations} \label{sec:sim_details:stop} 

Our topic-weighted BCMs can take a long time to reach their steady states.
We stop our simulations once we have a reasonable approximation of the limit opinion clusters by applying the following stopping criterion.
Consider the effective graph $G_\mathrm{eff}(t)$ and suppose it has $R$ opinion clusters (i.e., sets of nodes that make up the connected components of $G_\mathrm{eff}(t)$; see \cref{sec:theoretical:effective_graphs}), which we denote $S_1(t), \dots, S_R(t)$.
We terminate a simulation if the maximum difference in opinions within each opinion cluster for each topic is less than a tolerance value $\texttt{tol}$. That is 
\begin{equation}  \label{eq: stopping_criterion}
  \max_{1\leq r \leq R} \,
  \max_{i,j\in S_r(t)} \, \max_{1 \leq k \leq K} ~ |x_i^k(t) - x_j^k(t)|<\texttt{tol} \,.
\end{equation}
We call the time step that a simulation reaches
this stopping criterion the \emph{convergence time} of the simulation; we denote this time by $T_f$. We call the opinion clusters at time $T_f$ the \emph{final opinion clusters}.
For our topic-weighted BCMs, if two opinion vectors $\bm{x}_i$ and $\bm{x}_j$ satisfy $|x_i^k - x_j^k| < c$ for all topics $k$, then $d_\mathrm{max}(\bm{x}_i, \bm{x}_j) < c$.
By choosing a tolerance value that is less than $c$, we have that each pair of nodes $i$ and $j$ in the same final opinion cluster satisfies $d_\mathrm{max}(\bm{x}_i(T_f), \bm{x}_j(T_f)) < c$. Therefore, each pair of nodes in the same final opinion cluster is receptive to each other on all topics at time $T_f$.
For our topic-weighted HK model, we use a tolerance value of $1 \times 10^{-6}$; for our topic-weighted DW model, due to longer simulation times, we use a tolerance value of $0.01$. Both tolerance values are less than the smallest confidence bound that we examine in our simulations (namely, $c = 0.025$). 

We choose the stopping criterion in \eqref{eq: stopping_criterion} based on the convergence results discussed in \cref{sec:theoretical}; we now briefly summarize the relevant results.
The final opinion clusters in our simulations approximate the opinion clusters of the effective graph 
$G_\mathrm{eff}^*$ that corresponds to a limit $\bX^*$ in our BCMs (see \cref{sec:theoretical:effective_graphs}).
Recall (see \cref{thm: HK-eff-graph}) that for our topic-weighted HK model, the effective graph $G_\mathrm{eff}(t)$ converges to $G_\mathrm{eff}^*$ in the limit as time goes to infinity.
For our topic-weighted DW model, the effective graph $G_\mathrm{eff}(t)$ converges to $G_\mathrm{eff}^*$ almost surely in the limit as time goes to infinity
(see \cref{thm: DW-eff-graph}). 
Furthermore, $\bX^*$ is an absorbing state for our topic-weighted HK model (see \cref{thm: HK-absorbing}) and is almost surely an absorbing state for our topic-weighted DW model (see \cref{thm: DW-absorbing}).
When $\bX^*$ is an absorbing state, we have that 
(1) nodes $i$ and $j$ in the same opinion cluster of $G_\mathrm{eff}^*$ satisfy $d_\mathrm{max}(\bx_i^*, \bx_j^*) = 0$; and (2) adjacent nodes $i$ and $j$ in different opinion clusters of $G_\mathrm{eff}^*$ satisfy $d_\mathrm{max}(\bx_i^*, \bx_j^*) \geq c$.
When we reach our stopping criterion in a simulation, nodes $i$ and $j$ in the same final opinion cluster satisfy $d_\mathrm{max}(\bx_i(T_f), \bx_j(T_f)) < \texttt{tol}$. 
Adjacent nodes $i$ and $j$ that are in different final opinion clusters satisfy $d_\mathrm{max}(\bx_i(T_f), \bx_j(T_f)) \geq c$.
For a small tolerance value, the final opinion clusters in our simulations are a reasonable approximation of the opinion clusters of $G_\mathrm{eff}^*$. 
Through interactions with nodes in other final opinion clusters, it is possible for a pair of adjacent nodes in the same final opinion cluster to end up in different opinion clusters of $G_\mathrm{eff}^*$. It is also possible for a pair of adjacent nodes in different final opinion clusters to end up in the same opinion cluster of $G_\mathrm{eff}^*$. However, as we decrease our tolerance value in \eqref{eq: stopping_criterion}, the possibility of such scenarios decreases. 

For our topic-weighted DW model, the selection of a pair of nodes to interact and a topic on which to interact at each time step is random.
For each simulation of our topic-weighted DW model with $\omega \in (0, 1)$, we simultaneously run a control simulation.  
In a control simulation, we use the topic weight $\omega = 1$ (to give the baseline DW model) and use the same confidence bound $c$, set of initial opinions, and sequence of edges and topics for interactions as in the corresponding simulation of our topic-weighted DW model. 
If a simulation of our topic-weighted DW model reaches the stopping criterion (see \eqref{eq: stopping_criterion}) before the corresponding control simulation, we stop the topic-weighted DW simulation and continue the corresponding control simulation. Similarly, if a control simulation reaches the stopping criterion before the topic-weighted DW simulation, we continue the corresponding topic-weighted DW simulation.

\subsection{Characterizing opinion fragmentation
in our simulations}\label{sec:sim_details:characterizing_opinions}

Researchers have introduced various notions of opinion polarization and fragmentation~\cite{bramson2016} and have developed multiple methods to quantify them~\cite{bramson2016, musco2021, adams2022}. 
Intuitively, one simple way to quantify opinion fragmentation is to look at the number of opinion clusters, with a larger number of opinion clusters indicating more opinion fragmentation. However, this does not take into account the cluster sizes (i.e., the number of nodes in a cluster).
Suppose that there are two opinion clusters. If the two opinion clusters are the same size, then one can view the opinions in the system as being more fragmented than if one opinion cluster has a large majority of the nodes and the other opinion cluster has a small minority.
In our numerical simulations, there can be considerable variation in the sizes of the opinion clusters.
Therefore, it is important to consider both the numbers of opinion clusters and their sizes to quantify opinion fragmentation.

To quantify opinion fragmentation in our numerical simulations, we use a modification of the order parameter $Q$ defined by Wang et al.~\cite{wang2017_order_param}. We calculate\footnote{Wang et al.~\cite{wang2017_order_param} used the inequality $|x_i - x_j| \leq c$ in their calculation of $Q$ for the standard HK model. However, we use a strict inequality to be consistent with the strict inequality in the update rules \eqref{eq: HK_update_rule} and \eqref{eq: DW_update_rule}.} the order parameter $Q$ as 
\begin{equation}\label{eq:order_param}
  Q(t) = \frac{1}{M^2} \sum\limits_{i=1}^M \sum\limits_{j=1}^M \boldsymbol{1}_{d_\mathrm{max}(\bm{x}_i(t), \bm{x}_j(t)) < c} \,,
\end{equation}
where $\boldsymbol{1}_{d_\mathrm{max}(\bm{x}_i(t), \bm{x}_j(t)) < c}$ is an indicator function that equals 1 when $d_\mathrm{max}(\bm{x}_i(t), \bm{x}_j(t)) < c$ and equals $0$ otherwise. The order parameter $Q$ indicates the fraction of pairs of nodes that are receptive to each other on all topics. When $Q = 1$, all nodes are receptive to each other on all topics, and we say that nodes reach a \emph{consensus state} in the infinite-time limit, which is an opinion state in which all opinion vectors are equal. When $Q < 1$, we say that the nodes are in a \emph{fragmented state}. We say that smaller values of $Q$ indicate more \emph{opinion fragmentation}. For a complete graph, the order parameter gives the fraction of edges $(i,j) \in E$ that are in the effective graph. In a graph that is not complete, some nodes are not adjacent to each other, so the order parameter is more difficult to relate to the effective graph.

One can equivalently calculate the order parameter $Q(t)$ as a normalized square sum of the sizes of the opinion clusters. Suppose that at time $t$ there are $R$ opinion clusters, which we denote as $S_r(t)$ for $r \in \{1, \ldots, R\}$. One can calculate $Q(t)$ as
\begin{equation}\label{eq:order_param_squares}
    Q(t) = \frac{1}{M^2} \sum\limits_{r=1}^R |S_r(t)|^2 \,.
\end{equation}

For each of our simulations, we characterize opinions by calculating the \emph{final order parameter} $Q(T_f)$ at the time at which our simulation reaches the stopping criterion. 
Let $Q^*$ denote the $\emph{limit order parameter}$, which is the order parameter (see \eqref{eq:order_param}) calculated with the limit opinion vectors $\bm{x}_i^*$ for each node $i$. 
If $Q(T_f) = 1$, then there is a single final opinion cluster and all pairs of adjacent nodes are receptive to each other on all topics. In this scenario, nodes will compromise their opinions over time and eventually reach a consensus state, yielding a limit order parameter $Q^* = 1$. 
If $Q(T_f) < 1$, then there are at least two final opinion clusters. In this scenario, the final opinion clusters need not be the same as the opinion clusters of $G_\mathrm{eff}^*$, and it is possible that $Q(T_f) \neq  Q^*$.
However, for a small tolerance value, $Q(T_f)$ is a reasonable approximation of $Q^*$.

\subsection{Initial opinion distributions} \label{sec:sim_details:init_opinion_dist} 

In our simulations of our topic-weighted BCMs, we consider random initial opinions drawn from various probability distributions.
In studies on BCMs, it is relatively uncommon to consider initial opinions that arise from distributions other than uniform distributions (see discussion in \cref{sec:intro}).
We simulate our models with various initial opinion distributions to examine the effect of different initial opinion distributions on the opinion dynamics of the two topics.

To study our topic-weighted BCMs with independent initial opinions, we draw the initial opinions of the two topics independently and uniformly at random from $[0,1]$. This initial opinion distribution is a simple extension to the 1D standard HK and DW models (see \cref{sec:1D_HK,sec:1D_DW}).

To study our topic-weighted BCMs with interdependent initial opinions, we draw the initial opinions of each node independently from a wedge distribution, 
which has probability density function
\begin{equation}\label{eq:wedge_dist}
  f(x^1,x^2) = \delta\left(|2x^1-1|+x^2-1\right)
\end{equation}
if $x^1,x^2\in [0,1]$ and $0$ otherwise,
where $\delta$ is the Dirac delta function.
We show this wedge distribution in \cref{fig:wedge_dist}. 
The marginal opinion distributions for $x^1$ and $x^2$ (the opinions on topics 1 and 2, respectively) are uniform on $[0,1]$, which is identical to the independent 
case. However, $x^1$ and $x^2$ are highly correlated. To sample from this wedge distribution $f(x^1,x^2)$, we draw $x^1$ first uniformly at random from $[0,1]$ and then compute $x^2$ using $x^2=1-|2x^1-1|$. 

\begin{figure}[ht]
  \centering
  \includegraphics[width=0.28\textwidth]{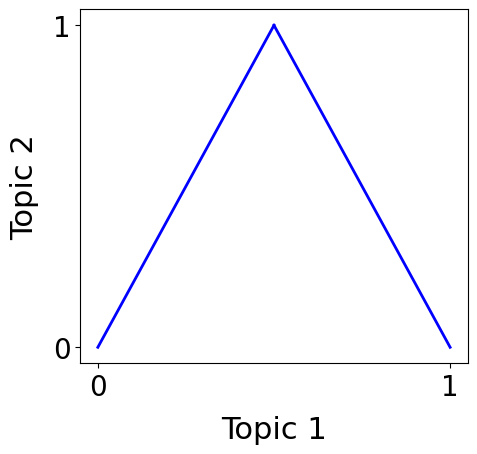}
  \caption{The probability density function $f(x^1,x^2)$ for the wedge distribution for the initial opinions of two topics. The probability density function $f$ has nontrivial mass on the blue curve.}
  \label{fig:wedge_dist} 
\end{figure}

We also examine our topic-weighted BCMs on truncated multivariate Gaussian distributions, for which we can control the correlation between the initial opinions of different topics.
Let $G(\sigma, \rho)$
denote a 2D truncated Gaussian distribution with mean $(0.5, 0.5)$ and the probability density function
\begin{align}
  f(x^1, x^2) \propto \exp \left\{ - \frac{1}{2\sigma^2(1-\rho^2)} \left[
    (x^1 - 0.5)^2 - 2\rho (x^1 - 0.5)(x^2 - 0.5) + (x^2 - 0.5)^2
  \right]\right\} \label{eq:gaussian_dist}
\end{align}
if $x^1,x^2\in [0,1]$ and $0$ otherwise.
Here, $x^1$ and $x^2$ are the opinions on topics 1 and 2, respectively, 
$\sigma$ is the standard deviation of the opinions on topics 1 and 2, and $\rho$ is the Pearson correlation between the opinions on topics 1 and 2.
In particular, we consider two truncated Gaussian distributions in our simulations: (1) $G(0.22, 0)$, where $x^1$ and $x^2$ have the same variance and are independent; and (2) $G(0.22, 0.8)$, where $x^1$ and $x^2$ have the same variance but are correlated. 
To draw opinions randomly from a truncated Gaussian distribution, we first draw randomly from a non-truncated Gaussian distribution. If the sample does not lie in $[0,1] \times [0,1]$, we reject it and resample using the same procedure. 
For our choice of $\sigma = 0.22$, we expect that 95\% of the random samples lie within the opinion space $[0,1] \times [0,1]$, which makes this simple rejection sampling method effective. 

\begin{figure}[ht]
  \centering
  \includegraphics[width=0.72\textwidth]{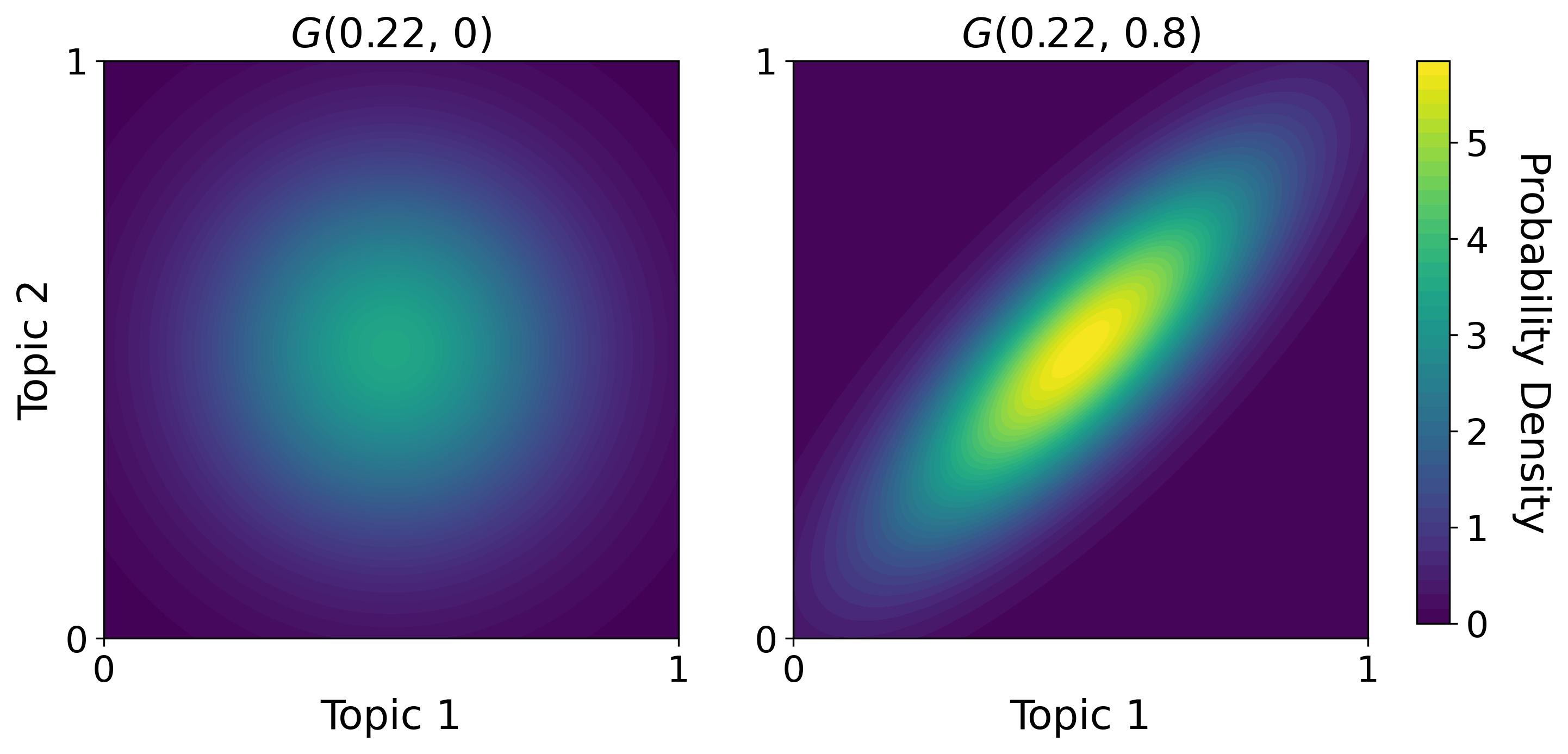}
  \caption{The probability density functions for the truncated Gaussian $G(0.22, 0)$ and $G(0.22, 0.8)$ distributions of initial opinions in our simulations of our topic-weighted BCMs.
  }
  \label{fig:gaussian_dist} 
\end{figure}

\subsection{Simulation specifications for models on complete graphs} \label{sec:sim_details:complete} 

We first simulate our topic-weighted BCMs on complete graphs. We simulate our topic-weighted HK model on a 2000-node complete graph. Our topic-weighted DW model is computationally more expensive than our topic-weighted HK model; we simulate it on a 500-node complete graph.
In our simulations, we examine uniform, wedge, $G(0.22, 0)$, and $G(0.22, 0.8)$ initial opinions.

We consider values of the confidence bound $c$ that allow us to examine the transition between opinion consensus and fragmentation.
We consider $c \in [0.025, 0.3]$ and $c \in [0.05, 0.4]$ for our topic-weighted HK simulations and our topic-weighted DW simulations, respectively. For both topic-weighted BCMs, we use $\omega \in [0.1, 1]$ in our simulations.

In our simulations, both the initial opinions and the selection of interacting node pairs and topics are random. To mitigate the impact of randomness, we repeat the simulations and compute averaged quantities.
We randomly generate 25 sets of initial opinions for each initial opinion distribution. We reuse those sets of initial opinions in our simulations 
with different values of the BCM parameters (namely, $c$ and $\omega$).

\subsection{Simulation specifications for our topic-weighted HK model on stochastic-block-model (SBM) graphs} \label{sec:sim_details:sbm} 

We use numerical simulations to investigate how the network community structure and community-dependent initial opinions affect opinion evolution. 
To do this, we simulate
our topic-weighted HK model on undirected two-community stochastic-block-model (SBM) networks \cite{holland1983_sbm,newman2018_book}. 
We do not run analogous simulations for our topic-weighted DW model because of its long simulation times.

We consider SBMs with a $2 \times 2$ block structure in which each block corresponds to an Erd\H{o}s--R\'enyi (ER) random graph $\mathcal{G}(N,p)$.
The ER graph $\mathcal{G}(N,p)$ has $N$ nodes and for each pair of distinct nodes, the edge between them exists in the graph with probability $p$, independently from every other edge \cite{gilbert1959_gnp}. 
When $p = 1$, every possible edge is in the graph and $\mathcal{G}(N,p)$ yields a complete graph of $N$ nodes.
To generate our SBMs with two communities, we partition the set of nodes into two subsets, \textrm{A} and \textrm{B}, with each consisting of half of the nodes in the network.
We define a symmetric edge-probability matrix
\begin{equation}
    P = \begin{bmatrix}
    P_\textrm{AA} & P_\textrm{AB} \\ P_\textrm{AB} & P_\textrm{BB}
    \end{bmatrix} \,,
\end{equation}
where $P_\textrm{AA}$ and $P_\textrm{BB}$ are the probabilities of an edge between two nodes within the sets \textrm{A} and \textrm{B}, respectively, and $P_\textrm{AB}$ is the probability of an edge between a node in set \textrm{A} and a node in set \textrm{B}. In our simulations, we consider 2000-node SBM graphs with $P_\textrm{AA} = P_\textrm{BB} = 1$ and $P_\textrm{AB} = 0.1$. These SBM graphs consist of two 1000-node cliques\footnote{A clique is a subgraph in which every pair of nodes is adjacent to each other.}, and a probability $P_\textrm{AB} = 0.1$ of an edge between two nodes in different cliques. 

For our simulations of our topic-weighted HK model on two-community SBM graphs, we consider both community-dependent and community-independent initial opinions. 
We draw the initial opinions from the uniform, wedge, and truncated Gaussian $G(0.22, 0)$ distributions described in \cref{sec:sim_details:init_opinion_dist}.

For our simulations with community-dependent initial opinions, nodes in community $\textrm{A}$ have initial opinions that lie in $[0,0.5] \times [0,1]$, and nodes in the community $\textrm{B}$ have initial opinions that lie in $[0.5, 1] \times [0,1]$. Thus, the two communities initially have split views on topic 1, but not on topic 2.
For the community-dependent uniform and wedge distributions, we select the initial opinions for topic 1 uniformly at random from $[0, 0.5]$ for nodes in community $\textrm{A}$ and from $[0.5, 1]$ for nodes in community $\textrm{B}$.
We then determine the corresponding opinions for topic 2 by uniformly random selection from $[0,1]$ for the uniform distribution and by computing $x^2=1-|2x^1-1|$ for the wedge distribution.
For the truncated Gaussian $G(0.22, 0.8)$ distribution, we start by drawing 2000 opinions, one for each node in our SBM graph.
We reject individual opinions and resample until we have exactly 1000 opinions that lie in $[0, 0.5] \times [0,1]$ for nodes in community $\textrm{A}$ and 1000 opinions that lie in $[0.5,1] \times [0,1]$ for nodes in community $\textrm{B}$.

For each of our simulations of our topic-weighted HK model with community-dependent initial opinions, we also run a corresponding simulation with community-independent initial opinions. This corresponding simulation uses the same two-community SBM graph, set of initial opinions, and BCM parameters. However, we shuffle the set of initial opinions and randomly assign them to nodes in the graph.
These simulations allow us to investigate the effect of having community-dependent initial opinions without introducing effects from varying the graph or set of initial opinions.

We randomly generate 5 two-community SBM graphs and 10 sets of initial opinions for each random graph. 
We reuse these 50 combinations of the random graph and set of initial opinions in our simulations with different values of the BCM parameters (namely, $c$ and $\omega$).
We use BCM parameter values $c \in [0.025, 0.5]$ and $\omega \in [0.1, 1]$ in our simulations.

\section{Results of our numerical simulations}\label{sec:numerical_results}

We now discuss the results\footnote{Our figures and simulation code are available at \url{https://gitlab.com/graceli1/topic-weighted-BCMs/}.} of our simulations of our topic-weighted BCMs model with various initial opinion distributions (see \cref{sec:sim_details:init_opinion_dist}) and values of the BCM parameters (namely, the confidence bound $c$ and topic weight $\omega$).

\subsection{Topic-weighted HK model}\label{sec:results_hk}

\subsubsection{A complete graph}\label{sec:results:hk_complete}

We now discuss the simulations of our topic-weighted HK model on a 2000-node complete graph for uniform, wedge, and two 
truncated Gaussian initial opinion distributions.
As we discussed in \cref{sec:sim_details:complete}, for each initial opinion distribution, we simulate our topic-weighted HK model using 25 distinct sets of initial opinions and compute the averaged order parameters.

In \cref{fig:HK_heatmap_Q}, we show the final order parameters $Q(T_f)$ (see \eqref{eq:order_param}) in our topic-weighted HK model for various initial opinion distributions. 
When $Q(T_f) = 1$, nodes eventually reach a consensus state (i.e., their opinions converge to the same limit opinion). 
Smaller values of $Q(T_f)$ indicate more opinion fragmentation than larger values of $Q(T_f)$. 
For fixed $\omega$ and a fixed initial opinion distribution, we observe smaller values of $Q(T_f)$ (and correspondingly more opinion fragmentation) as we decrease $c$. 
For the uniform and Gaussian initial opinion distributions, for fixed $c$, we tend to observe less opinion fragmentation as we decrease $\omega$. This is especially noticeable near the transition between consensus and fragmentation.
By contrast, for the wedge distribution and fixed $c$, varying $\omega$ typically has little effect on $Q(T_f)$. An exception is that when $c = 0.225$ (near the transition between consensus and fragmentation), $Q(T_f)$ tends to increase as we decrease $\omega$. 

\begin{figure}[htb]
  \centering
  \includegraphics[width=0.9\textwidth]{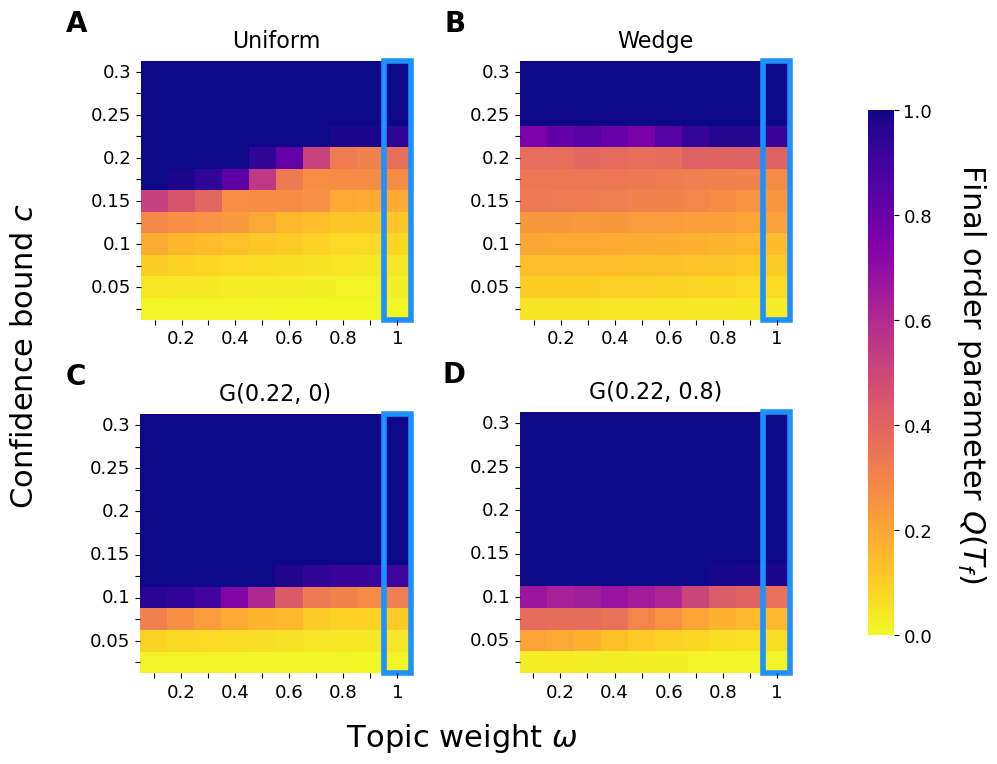}
  \caption{Final order parameters $Q(T_f)$ (see \eqref{eq:order_param}) in our topic-weighted HK model on a 2000-node complete graph with various initial opinion distributions and values of the BCM parameters $c$ and $\omega$. When $\omega = 1$, we have the baseline HK model, and in this figure and subsequent similar figures, we have indicated these simulations with a blue border. In each heatmap, each cell shows the mean value of 25 simulations.
  }
  \label{fig:HK_heatmap_Q} 
\end{figure}

In our simulations, the choice of initial opinion distribution significantly influences whether our topic-weighted HK model converges to a consensus state or a fragmented state.
For fixed $\omega$, the transition between consensus and fragmentation occurs at smaller values of $c$ for the two Gaussian initial opinion distributions than for the uniform and wedge initial opinion distributions. 
Intuitively, it makes sense that the two Gaussian distributions promote less fragmentation than the uniform and wedge distributions. This is because the initial opinions for the Gaussian distributions are more concentrated near the center (see \cref{fig:gaussian_dist}) and less concentrated near the boundary of the opinion space $[0,1]\times [0,1]$.
For the uniform and Gaussian $G(0.22, 0)$ distributions, as we decrease $\omega$, the value of $c$ at which the transition between consensus and fragmentation occurs also decreases. 
For the uniform distribution and $c = 0.2$, our topic-weighted HK model with $\omega \leq 0.4$ reaches consensus, while the baseline HK model has opinion fragmentation.

To investigate the impact of the distribution of initial opinions on the distribution of final opinions, we plot the trajectories of the opinions of nodes in our simulations. In our opinion space $[0,1] \times [0,1]$, for each node, we plot the path of its opinion; that is, we plot how its opinion changes from its initial opinion to its opinion at convergence time $T_f$. 
In \cref{fig:opinion_trajectory_HK_complete_uniform}, we show some plots of the opinion trajectories for simulations with uniform initial opinions and $c = 0.1$ for both our topic-weighted HK model with $\omega = 0.1$ and the baseline HK model.
In the simulation of the baseline HK model (see \cref{fig:opinion_trajectory_HK_complete_uniform}B), for each final opinion cluster, the initial opinions of the nodes in that opinion cluster appear to visually form a rectangle containing the corresponding the mean final opinion; the opinion space can be split into non-overlapping rectangular regions that cover the space and nodes with initial opinions in the same rectangular region are in the same final opinion cluster.  
By contrast, in the simulation of our topic-weighted HK model with $\omega = 0.1$ (see \cref{fig:opinion_trajectory_HK_complete_uniform}A), for each final opinion cluster, the initial opinions of nodes in that opinion cluster do not form rectangular shapes.
The simulation of our topic-weighted HK model with $\omega = 0.1$ has $Q(T_f) \approx 0.1956$ and there is less opinion fragmentation than the simulation of the baseline HK model, which has $Q(T_f) \approx 0.0665$.
Recall (see \cref{fig:region_of_absorption}) that in 2D, the baseline HK model has square regions of receptiveness and our topic-weighted HK model with $\omega \in (0, 1)$ has 8-sided regions of receptiveness. For fixed $c$, as we decrease $\omega$, the area of the regions of receptiveness increases. We hypothesize that the square shape of the regions of receptiveness for the baseline HK model yields the rectangular regions in \cref{fig:opinion_trajectory_HK_complete_uniform}B. 
The larger size and non-square shape of the regions of receptiveness for our topic-weighted HK model with $\omega = 0.1$ result in nodes with initial opinions in larger and more irregular regions of opinion space being pulled into the same final opinion cluster, yielding less opinion fragmentation than the baseline model.
We hypothesize that a similar effect applies to the truncated Gaussian distributions, causing smaller values of $\omega$ in our topic-weighted HK model to promote less opinion fragmentation. We show some plots of opinion trajectories for the truncated Gaussian distributions of initial opinions $G(0.22, 0)$ (see \cref{fig:opinion_trajectory_HK_complete_gaussian_corr0}) and $G(0.22, 0.8)$ (see \cref{fig:opinion_trajectory_HK_complete_gaussian_corr0.8})
in \cref{sec:appendix:HK_complete_gaussian}.
\begin{figure}[htb]
    \centering
    \includegraphics[width=0.75\textwidth]{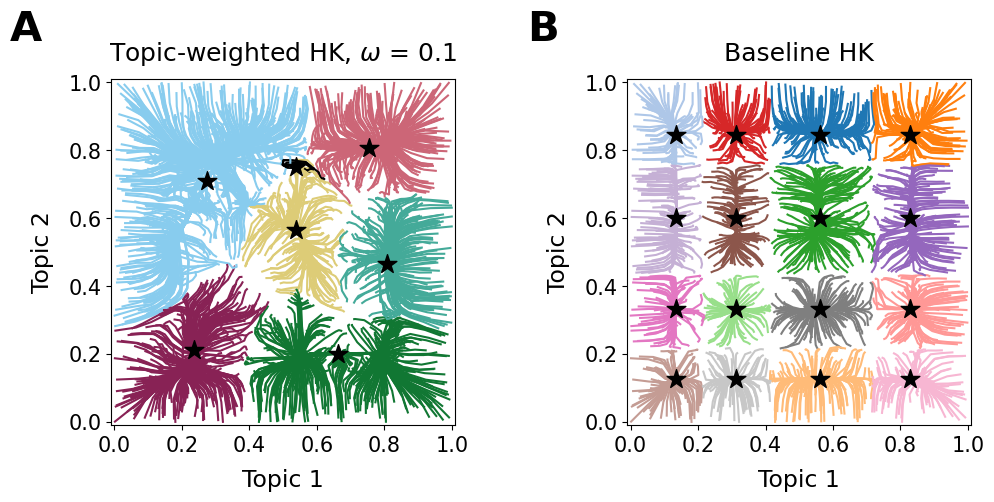}
    \caption{Opinion trajectory plots for simulations on a 2000-node complete graph of (A) our topic-weighted HK model with $\omega = 0.1$ and (B) the baseline HK model. Both simulations use $c = 0.1$ and the same set of initial opinions from a uniform distribution. 
    For this figure and subsequent figures of opinion trajectories, nodes in the same final opinion cluster have the same color opinion trajectory. For each final opinion cluster, a black star marks the mean of the final opinions of nodes in that final opinion cluster.
    }
    \label{fig:opinion_trajectory_HK_complete_uniform}
\end{figure}

In \cref{fig:opinion_trajectory_HK_complete_wedge}, we show the opinion trajectories for simulations with wedge initial opinions and $c = 0.225$ for both our topic-weighted HK model with $\omega = 0.1$ and the baseline HK model.
For the wedge initial opinion distribution and fixed $c$, varying $\omega$ appears to have little effect on the final order parameter.
This is likely due to the shape of the wedge distribution (see \eqref{eq:wedge_dist}). 
As we discussed previously, for uniform and Gaussian initial opinion distributions, fixing $c$ and decreasing $\omega$ increase the area of the regions of receptiveness and result in less opinion fragmentation. 
For small $c$, adjacent nodes that are near the ends of the two different ``legs'' of the wedge distribution (specifically, adjacent nodes $i$ and $j$ with $x_i^1(0) \leq 0.2$ and $x_j^1(0) \geq 0.8$) are unreceptive to each other for all values of $\omega$.
We hypothesize that for small $c$, such adjacent nodes are unable to compromise in our topic-weighted HK model for any value of $\omega$, and we do not observe less opinion fragmentation as we decrease $\omega$.
\begin{figure}[htb]
    \centering
    \includegraphics[width=0.75\columnwidth]{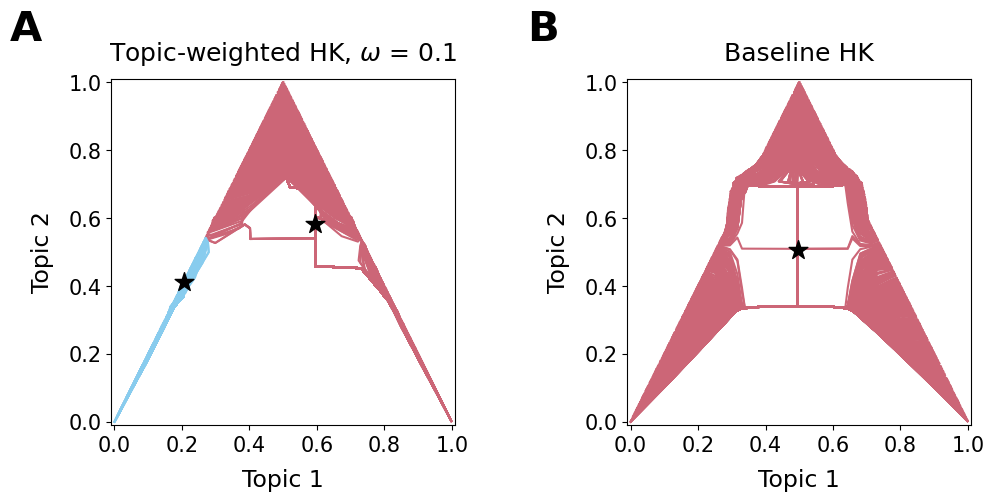}
    \caption{Opinion trajectory plots for simulations on a 2000-node complete graph of (A) our topic-weighted HK model with $\omega = 0.1$ and (B) the baseline HK model. Both simulations use $c = 0.225$ and the same set of initial opinions from a wedge distribution.
    }
    \label{fig:opinion_trajectory_HK_complete_wedge}
\end{figure}

For the wedge initial opinion distribution (see \cref{fig:HK_heatmap_Q}B), when $c = 0.225$, which is near the transition between opinion fragmentation and consensus, small values of $\omega$ yield more opinion fragmentation. This increase in fragmentation for small $\omega$ is the opposite of the trend that we observe for the uniform and Gaussian distributions.
When a simulation with wedge initial opinions reaches consensus, there is a single final opinion cluster with corresponding mean final opinion near $(0.5, 0.5)$ at the center of the opinion space (see \cref{fig:opinion_trajectory_HK_complete_wedge}B). For the baseline HK model with $c = 0.225$, most simulations reach consensus, but a few have two opinion clusters; when there are two opinion clusters, they both have corresponding mean final opinions in which either the topic 1 or topic 2 opinion is near 0.5 (that is, $x^1 \approx 0.5$ or $x^2 \approx 0.5$). 
For our topic-weighted HK model with $c = 0.225$ and small $\omega$, most simulations have two or three final opinion clusters and therefore there is more opinion fragmentation. When there are two final opinion clusters, as in \cref{fig:opinion_trajectory_HK_complete_wedge}A, neither corresponding final opinion has its topic 1 or topic 2 opinion near 0.5 (that is, neither $x^1 \approx 0.5$ nor $x^2 \approx 0.5$). We hypothesize that near the transition between consensus and fragmentation, for our topic-weighted HK model with small $\omega$, it is more difficult than in the baseline model for nodes with initial opinions near the end of one of the ``legs'' of the wedge distribution to reach a consensus opinion.

\subsubsection{Two-community SBM graphs}\label{sec:results:hk_sbm}

We now discuss the simulations of our topic-weighted HK model on two-community SBM random graphs.
As we discussed in \cref{sec:sim_details:sbm}, our SBM random graphs consist of two communities with 1000 nodes each. 
For our simulations, we generate five random SBM graphs. For each graph and each initial opinion distribution, we use 10 distinct sets of initial opinions in Monte Carlo simulations of our topic-weighted HK model. For each combination of random SBM graph, set of initial opinions, and set of BCM parameters, we simulate two scenarios: one with community-dependent assignment of initial opinions and one with random assignment of initial opinions.

For each of the uniform, wedge, and truncated Gaussian $G(0.22, 0)$ initial opinion distributions and fixed BCM parameters, when initial opinions are randomly assigned to nodes, the final order parameter $Q(T_f)$ (see \eqref{eq:order_param}) for the two-community SBM graphs is similar to $Q(T_f)$ for a complete graph. 
The minimum value of $c$ for our simulations to always reach consensus is larger when we assign initial opinions to nodes based on their community, compared to when we randomly assign initial opinions to nodes.
It appears that the two-community structure of our SBM graphs has little effect on the amount of opinion fragmentation when initial opinions are independent of community, but tends to increase opinion fragmentation when we have community-dependent initial opinions in which the communities are ``split'' on their initial opinions on one topic (specifically, nodes in one community have $x^1(0) \in [0, 0.5]$ and nodes in the other community have $x^1(0) \in [0.5, 1]$; see \cref{sec:sim_details:sbm}).
For all our simulations with community-dependent initial opinions, when consensus is not reached, nodes in the same final opinion cluster also belong to the same community.

In \cref{fig:HK_heatmap_2_community_uniform}, we show the final order parameters $Q(T_f)$ of our simulations of our topic-weighted HK model on two-community SBM graphs with uniform initial opinions.
For uniform initial opinions, our simulations on SBM graphs, both those with community-dependent initial opinions and those with randomly-assigned initial opinions, have the same general trends as our simulations on complete graph;
namely, we tend to observe more opinion fragmentation as we either (1) decrease $c$ for fixed $\omega$ or (2) increase $\omega$ for fixed $c$.
For fixed $\omega$, a larger value of $c$ is required for simulations to always reach consensus for the community-dependent initial opinions (see \cref{fig:HK_heatmap_2_community_uniform}A) than for randomly-assigned initial opinions (see \cref{fig:HK_heatmap_2_community_uniform}B). 

\begin{figure}[htb]
  \centering
  \includegraphics[width=0.8\textwidth]{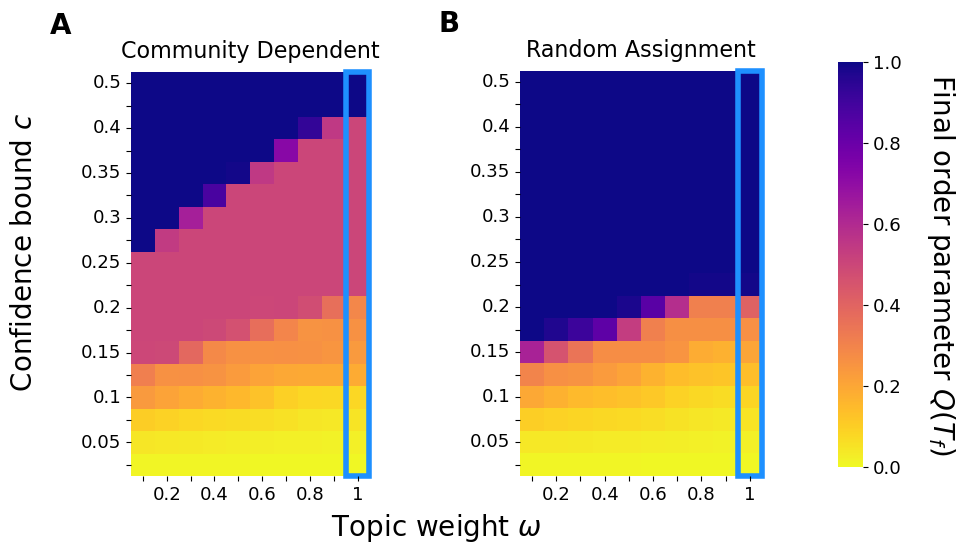}
  \caption{Final order parameters $Q(T_f)$ (see \eqref{eq:order_param}) in simulations of our topic-weighted HK model on two-community SBM graphs with uniform initial opinions for various values of the BCM parameters $c$ and $\omega$. We simulate our model with (A) community-dependent assignment of initial opinions and (B) random assignment of initial opinions.
  In each heatmap, each cell shows the mean value of 50 simulations.
  }
  \label{fig:HK_heatmap_2_community_uniform} 
\end{figure}

In \cref{fig:HK_SBM_opinion_trajectory_uniform}, we show the opinion trajectories for two simulations of our topic-weighted HK model with $\omega = 0.1$ and $c = 0.1$ and uniform initial opinions. 
In \cref{fig:HK_SBM_opinion_trajectory_uniform}A, each final opinion cluster is contained within one community, and each community can be split into three corresponding final opinion clusters. By contrast (see \cref{fig:HK_SBM_opinion_trajectory_uniform}B), when initial opinions are randomly assigned to nodes, final opinion clusters can consist of nodes from each community. 
For community-dependent initial opinions and intermediate values of $c$, we observe that $Q(T_f) = 0.5$ (see \cref{fig:HK_heatmap_2_community_uniform}). In this situation, there are two final opinion clusters; each final opinion cluster contains all the nodes from a single community. 

\begin{figure}[htb]
  \centering
  \includegraphics[width=0.75\textwidth]{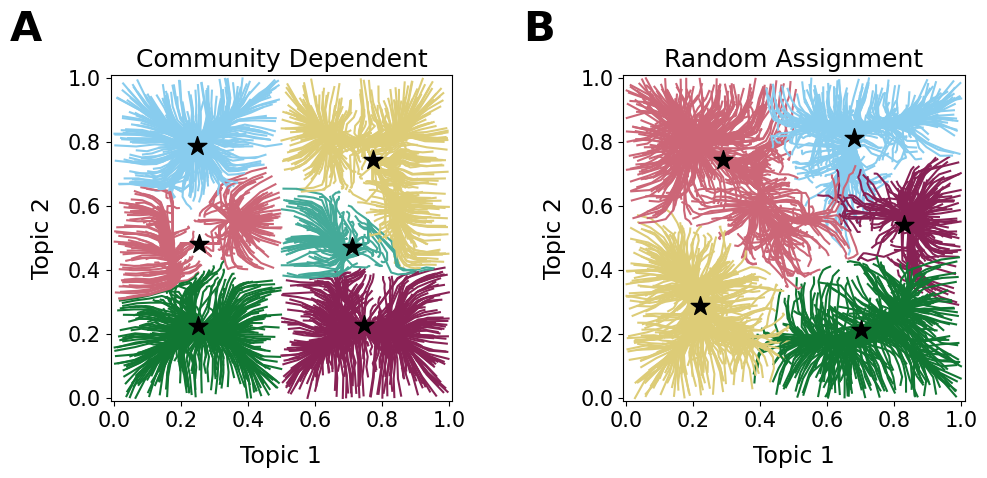}
  \caption{Opinion trajectory plots for simulations of our topic-weighted HK model with $\omega = 0.1$ and $c = 0.1$ on a two-community SBM graph. Both simulations use the same set of uniform initial opinions.
  One simulation (A) has initial opinions assigned to nodes based on their community and the other simulation (B) has initial opinions assigned to nodes randomly.
  }
  \label{fig:HK_SBM_opinion_trajectory_uniform} 
\end{figure}

Our topic-weighted HK simulations on SBM graphs with $G(0.22, 0)$ initial opinions, have similar trends as our simulations with uniform initial opinions.
In \cref{sec:appendix:HK_sbm_gaussian}, we show some plots and further discuss our results for simulations with $G(0.22, 0)$ initial opinions.

\begin{figure}[htbp!]
  \centering
  \includegraphics[width=0.8\textwidth]{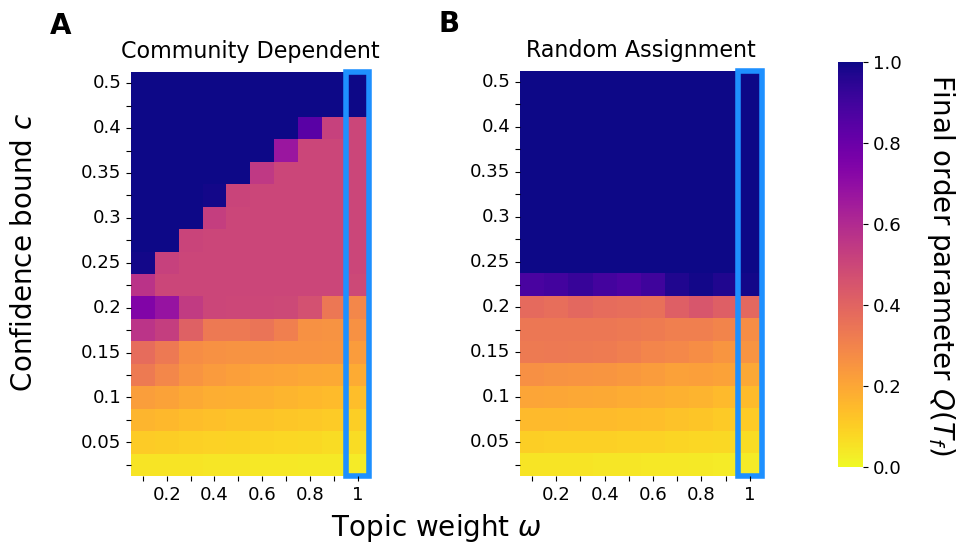}
  \caption{Final order parameters $Q(T_f)$ (see \eqref{eq:order_param}) in simulations of our topic-weighted HK model on two-community SBM graphs with wedge initial opinions for various values of the BCM parameters $c$ and $\omega$. We simulate our model with (A) community-dependent assignment of initial opinions and (B) random assignment of initial opinions.
  }
  \label{fig:HK_heatmap_2_community_wedge} 
\end{figure}

In \cref{fig:HK_heatmap_2_community_wedge}, we show the final order parameters $Q(T_f)$ of our simulations of our topic-weighted HK model on two-community SBM graphs with wedge initial opinions. Unlike for the community-dependent uniform and $G(0.22, 0)$ initial opinions, for the community-dependent wedge initial opinions, for small $\omega$, as we increase $c$, the final order parameter $Q(T_f)$ does not always increase or stay the same. For $\omega \in \{ 0.1, 0.2, 0.3\}$, we observe that $Q(T_f)$ when $c = 0.2$ is smaller than $Q(T_f)$ when either $c = 0.175$ or $c = 0.225$. 
In \cref{fig:HK_SBM_opinion_trajectory_wedge}, we show two opinion trajectory plots for simulations of our topic-weighted HK model with $\omega = 0.1$. Both simulations use the same community-dependent wedge initial opinions. In both simulations, the mean opinions associated with the two final opinion clusters are just outside of each other's regions of receptiveness. One can imagine that small perturbations of the initial opinions could result in a simulation that reaches consensus. 
In the two simulations shown, the opinions of the two final opinion clusters are closer to each other's region of absorption for the simulation with $c = 0.2$ (see \cref{fig:HK_SBM_opinion_trajectory_wedge}A) than for the simulation with $c = 0.225$ (see \cref{fig:HK_SBM_opinion_trajectory_wedge}B). Therefore, it makes sense that when $\omega \leq 0.3$, we observe that more simulations of our topic-weighted HK model with $c = 0.2$ reach consensus than simulations of our model with $c = 0.225$. We hypothesize that we only observe this behavior for the wedge initial opinions and not the uniform or Gaussian initial opinions because the shape of the wedge distribution causes the mean opinions of the final opinion clusters to be close to each other's regions of absorption when there are two final opinion clusters.

\begin{figure}[htbp]
  \centering
  \includegraphics[width=0.75\textwidth]{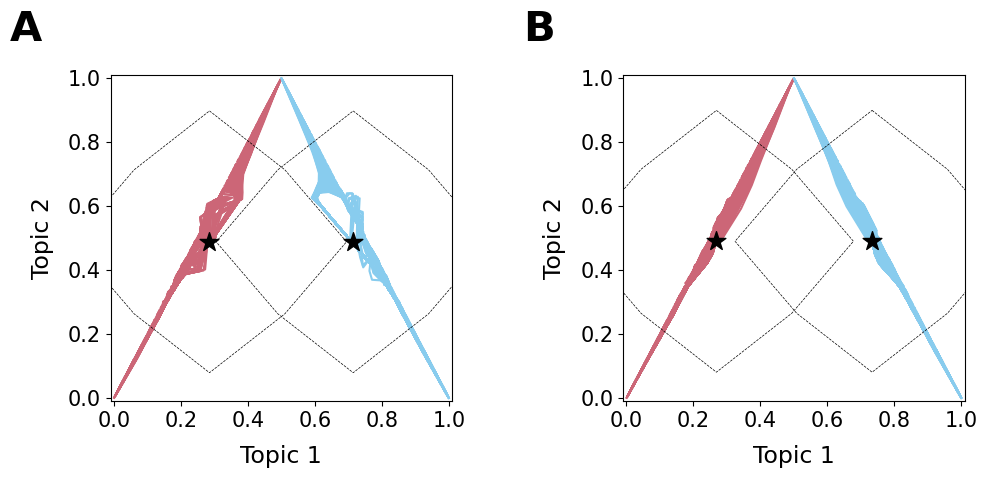}
  \caption{Opinion trajectory plots for simulations of our topic-weighted HK model on a two-community SBM graph with $\omega = 0.1$ and (A) $c = 0.2$ and (B) $c = 0.225$. Both simulations use the same set of wedge initial opinions that are assigned to nodes based on their community. The polygons indicated with the dotted lines show the region of receptiveness for the mean final opinion of each final opinion cluster.
  }
  \label{fig:HK_SBM_opinion_trajectory_wedge} 
\end{figure}

\subsection{Topic-weighted DW model}\label{sec:results_dw}

We now discuss the results of our simulations of our topic-weighted DW model on a 500-node complete graph for uniform, wedge, and two truncated Gaussian initial opinion distributions.
As we discussed in \cref{sec:sim_details:complete}, for each initial opinion distribution, we use 25 distinct sets of initial opinions in Monte Carlo simulations of our topic-weighted DW model.

In \cref{fig:DW_heatmap_Q}, we show the final order parameters $Q(T_f)$ (see \eqref{eq:order_param}) in our simulations of our topic-weighted DW model for various initial opinion distributions.
Overall, we observe similar trends for $Q(T_f)$ for our topic-weighted DW model as we did for our topic-weighted HK model.
In this section, we discuss our observations of $Q(T_f)$ for our topic-weighted DW model. 
In \cref{sec:appendix:DW_results}, we show and discuss some plots of opinion trajectories.

For the uniform and Gaussian initial opinion distributions, we observe similar trends for our topic-weighted DW model as we did for our topic-weighted HK model.
One trend is that, for fixed confidence bound $c$ and a fixed initial opinion distribution, we tend to observe less opinion fragmentation as we decrease the topic weight $\omega$. 
Another trend is that, as we decrease $\omega$, we observe a decrease in the value of $c$ at which the transition between consensus and fragmentation occurs.
For our topic-weighted DW model with fixed $\omega$ and fixed $c$, the Gaussian initial opinion distributions tend to have less opinion fragmentation than the uniform initial opinion distribution. However, for fixed $\omega$, the value of $c$ above which all simulations reach a consensus (i.e., $Q(T_f) = 1$ for all simulations) appears similar for the uniform and Gaussian initial opinion distributions. 
For our simulations of our topic-weighted DW model, the transition between consensus and fragmentation is not as clear as for our simulations of our topic-weighted HK model (see \cref{fig:HK_heatmap_Q}). We hypothesize that finite-size effects from the 500-node complete graph and randomness from selecting pairs of nodes to interact for our topic-weighted DW simulations may contribute to this observation.

As was the case for our topic-weighted HK model, for our topic-weighted DW model, $\omega$ appears to have little effect on the opinion fragmentation for the wedge initial opinion distribution. Specifically, for the wedge distribution with fixed $c$, we typically observe little effect on the $Q(T_f)$ when we vary $\omega$. 
An exception is that for $c = 0.25$, there appears to be noticeably less opinion fragmentation for $0.4 \leq \omega \leq 0.6$ than that for other values of $\omega$. 
For the wedge distribution, when $c = 0.25$, simulations typically result in 2--5 final opinion clusters. It appears that compared to other values of $\omega$, for $0.4 \leq \omega \leq 0.6$, there are more simulations in which there are at least 99\% of nodes in the largest final opinion, so nodes either reach a consensus or almost reach a consensus. 
This results in a smaller mean $Q(T_f)$ for $0.4 \leq \omega \leq 0.6$ than for other values of $\omega$. 
As we discuss further in \cref{sec:appendix:DW_results}, the shape of the wedge distribution and randomness in selecting pairs of nodes to interact and a topic to interact on in the topic-weighted DW model may contribute to this observation.

\begin{figure}[htb]
  \centering
  \includegraphics[width=0.95\textwidth]{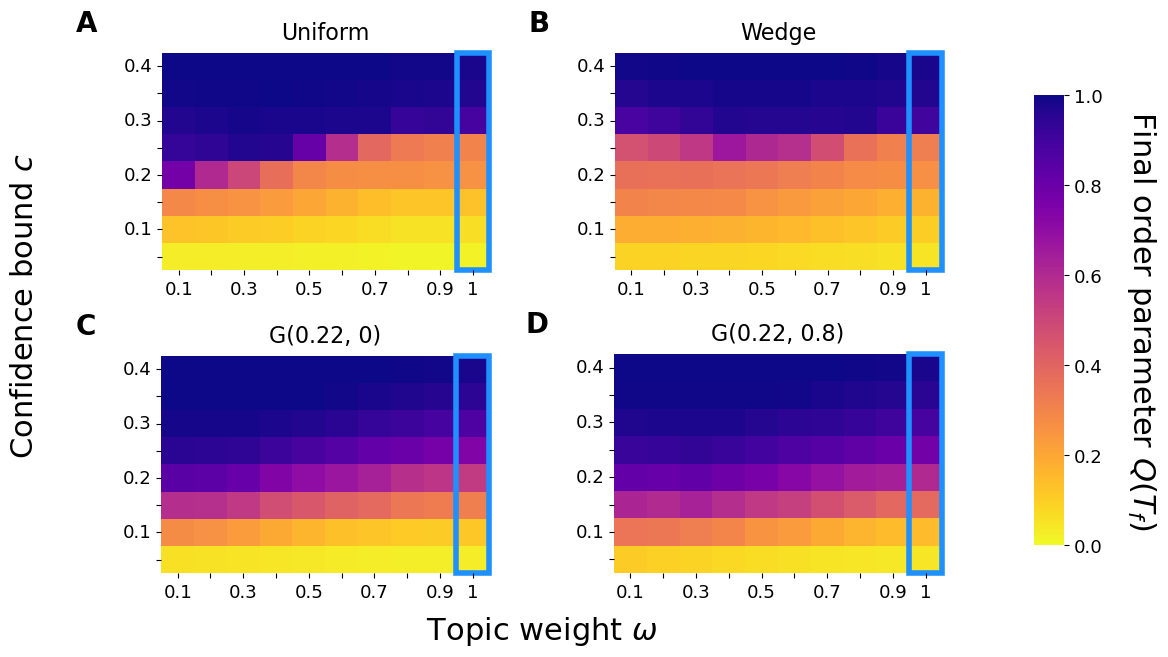}
  \caption{Final order parameters $Q(T_f)$ (see \eqref{eq:order_param}) in simulations of our topic-weighted DW model on a 500-node complete graph with various initial opinion distributions and various values of the BCM parameters $c$ and $\omega$. In each heatmap, for $\omega \leq 0.9$, each cell shows the mean value of 25 simulations, and for $\omega = 1$, each cell shows the mean value of 225 simulations.
  }
  \label{fig:DW_heatmap_Q} 
\end{figure}

\section{Conclusions and discussion}
\label{sec:conclusions}
In this paper, we introduced two bounded-confidence models (BCMs) of multidimensional opinions: a synchronous model extending the Hegselmann--Krause (HK) model and an asynchronous model extending the Deffuant--Weisbuch (DW) model.
We interpret the multidimensional opinion vectors in our models as opinions on multiple related topics, and we defined topic-weighted discordance functions to quantify the distance between opinion vectors. 
These topic-dependent discordance functions depend not only on the opinions for a specific topic $k$ but also incorporate the influence of other topics. 
For each topic $k$, the discordance function $d_k$ calculates a weighted average of the opinion differences across all topics, with a topic weight parameter $\omega$ determining how much more weight is assigned to the opinion difference in topic $k$ than the opinion difference in other topics.

We studied our topic-weighted BCMs by examining the convergence and steady-state properties of the dynamics. 
We proved that the limit of the opinion vectors in our BCMs exists as time goes to infinity, and the limit opinions of two adjacent nodes are either the same or separated in the sense that their topic-weighted discordances are greater than or equal to the confidence bound.
To characterize the separation of distinct limit opinions, we analyzed regions of receptiveness to rigorously quantify the neighborhood in the opinion space of a steady-state opinion cluster where no other opinion cluster exists. 
We also defined time-dependent effective graphs and showed their convergence to the effective graphs associated with the steady-state opinions.
We used these effective graphs to establish the stopping criterion for our numerical simulations.

In addition to investigating the theoretical aspects of our topic-weighted BCMs, we numerically simulated our topic-weighted BCMs with initial opinions sampled from various distributions, including both independent and interdependent initial opinions distributions on different topics.
We demonstrated that the choice of initial opinion distribution significantly influences the degree of opinion fragmentation (quantified by the order parameter $Q$) and the BCM parameters --- specifically, the confidence bound $c$ and topic weight $\omega$ --- at which the transition between consensus and fragmentation occurs.
For both our topic-weighted BCMs, for uniform and truncated Gaussian initial opinion distributions with a fixed $c$, decreasing $\omega$ tends to reduce opinion fragmentation, likely due to the expansion of the regions of receptiveness as $\omega$ decreases. 
However, for the wedge distribution, varying $\omega$ for fixed $c$ has little impact on opinion fragmentation, indicating that the effect of $\omega$ depends on the initial opinion distribution.
In the case of our topic-weighted HK model applied to two-community SBM graphs, the graph structure shows minimal influence on opinion fragmentation when initial opinions are independent of the community. However, fragmentation increases when initial opinions are correlated with community membership.
Our simulations highlight the substantial impact of initial opinion distributions on the behavior of our topic-weighted BCMs and baseline BCMs. 
Given that most research on BCMs only considers uniform initial opinions, we encourage researchers studying BCMs and other models of opinion dynamics with continuous-valued opinions to consider a wider range of initial opinion distributions.

In our numerical simulations of our topic-weighted BCMs, we observed that the shape of the region of receptiveness, which is determined by $\omega$ in our models, can affect the amount of opinion fragmentation and the qualitative behavior of the opinion trajectories.
Most prior studies on BCMs with multidimensional opinions consider Euclidean distance \cite{lorenz2006_multi, li2017_ja_vector,fortunato2005_vector, bhattacharyya2013, etesami2013_hk_convergence, hegselmann2019_hk_d_dimension_convergence, brooks2020}. 
The Euclidean norm yields spherical regions of receptiveness, which differ from the polytope-shaped regions of receptiveness in our models.
Multidimensional opinion models have a variety of choices for modeling opinion distance, which yield different regions of receptiveness. As an area of future work, it is of interest to compare the opinion dynamics of BCMs with different opinion distances and regions of receptiveness.

Our agent-based approach in modeling multidimensional opinion dynamics can also be translated into density-based opinion models characterized by an opinion density function, which describes the distribution of opinions across an entire population. 
(See \cite{ben2003bifurcations,chu2023density} for previous work on density-based opinion models.)
It is also beneficial to incorporate heterogeneous features into the topic-weighted BCMs. For instance, one could introduce topic heterogeneity into the discordance function, allowing certain topics to carry more influence than others, thereby capturing the unequal impact of topics in real-world scenarios. Additionally, one can extend our topic-weighted BCMs to multiplex networks, where individuals exchange opinions on different topics through social interactions across distinct layers of the network. 

\appendix

\section{Additional results of our topic-weighted HK simulations}

\subsection{A complete graph with truncated Gaussian initial opinions}
\label{sec:appendix:HK_complete_gaussian}

In this section, we show some plots of the opinion trajectories of 
our topic-weighted HK model on a 2000-node complete graph with initial opinions drawn from truncated Gaussian $G(0.22, 0)$ and $G(0.22, 0.8)$ distributions.

In \cref{fig:opinion_trajectory_HK_complete_gaussian_corr0}, we show some plots of the opinion trajectories for simulations with Gaussian $G(0.22, 0)$ initial opinions and confidence bound $c = 0.075$ for both our topic-weighted HK model with topic weight $\omega = 0.1$ and the baseline HK model. 
As was the case for uniform initial opinions (see \cref{fig:opinion_trajectory_HK_complete_uniform}), for the Gaussian initial opinions, simulations of the baseline HK model (see \cref{fig:opinion_trajectory_HK_complete_gaussian_corr0}B) tend to result in the initial opinions of the nodes in each final opinion cluster appearing to form a rectangular region containing the corresponding mean final opinion. 
Simulations of our topic-weighted HK model with $\omega = 0.1$ (see \cref{fig:opinion_trajectory_HK_complete_gaussian_corr0}A) tend to result in the initial opinions of nodes in each final opinion cluster forming non-rectangular shapes.
We observe that nodes with initial opinions that are away from the distribution mean at $(0.5, 0.5)$ tend to have opinions that move towards $(0.5, 0.5)$ over time. 
In \cref{fig:opinion_trajectory_HK_complete_gaussian_corr0}, for final opinion clusters with corresponding mean final opinion not near $(0.5, 0.5)$, the mean final opinion is near the boundary of the region formed by the corresponding initial opinions. 
Intuitively, this makes sense because the Gaussian $G(0.22, 0)$ distribution has a larger probability of initial opinions near the distribution mean at $(0.5, 0.5)$. If there are more nodes with initial opinions near the distribution mean, then these nodes have more weight in the HK update rule (see \eqref{eq: HK_update_rule}) and pull the opinions of other nodes towards the distribution mean.

\begin{figure}[tbh]
    \centering
    \includegraphics[width=0.75\textwidth]{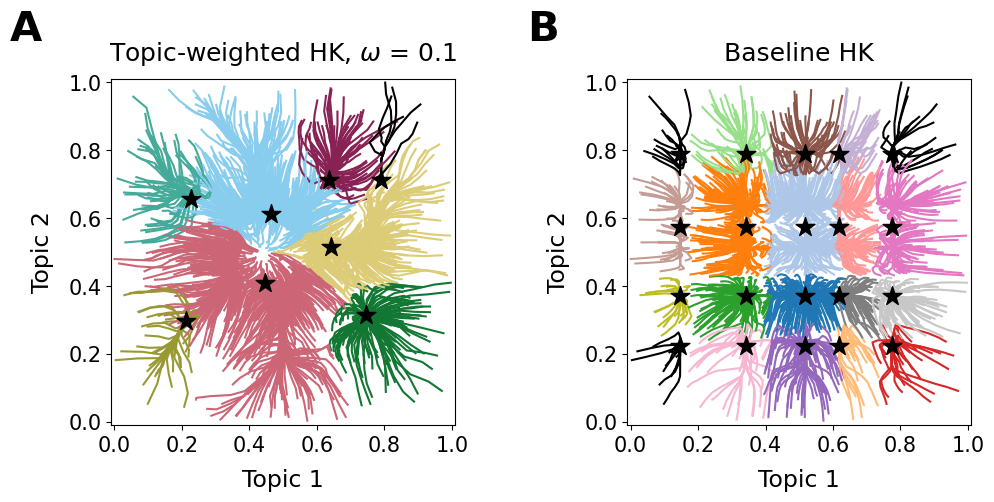}
    \caption{Opinion trajectory plots for simulations on a 2000-node complete graph of (A) our topic-weighted HK model with $\omega = 0.1$ and (B) the baseline HK model. Both simulations use $c = 0.075$ and the same set of initial opinions from a truncated Gaussian $G(0.22, 0)$ distribution.
    }
    \label{fig:opinion_trajectory_HK_complete_gaussian_corr0}
\end{figure}

In \cref{fig:opinion_trajectory_HK_complete_gaussian_corr0.8}, we show some plots of the opinion trajectories for simulations with Gaussian $G(0.22, 0.8)$ initial opinions and $c = 0.075$ for both our topic-weighted HK model with $\omega = 0.1$ and the baseline HK model.
As was the case for uniform and $G(0.22, 0)$
initial opinions, we tend to observe rectangular and non-rectangular regions formed by the final opinion clusters in the opinion-trajectory plots for the baseline and topic-weighted HK models, respectively.
In the simulation of our topic-weighted HK model with $\omega = 0.1$ (see \cref{fig:opinion_trajectory_HK_complete_gaussian_corr0}A), there are two nodes that are each in a final opinions cluster consisting of just themselves; they have final opinions near $(0.33, 0.52)$ and $(0.91, 0.37)$. These two nodes have initial opinions away from the skew axis (i.e., the line $x^1 = x^2$) of the $G(0.22, 0.8)$ distribution and quickly become unreceptive to all other nodes. 
There are four other final opinion clusters that have the 1,998 other nodes in the graph. These four final opinion clusters have mean final opinions near the skew axis of the $G(0.22, 0.8)$ distribution. Which final opinion cluster a node belongs to appears to relate to which point on the skew axis their initial opinion is closest to. 

\begin{figure}[htbp!]
    \centering
    \includegraphics[width=0.75\textwidth]{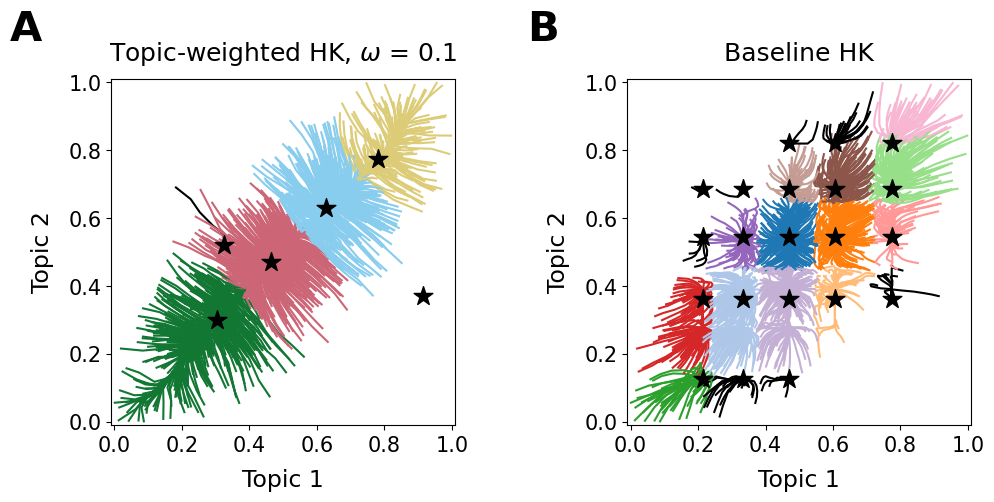}
    \caption{Opinion trajectory plots for simulations on a 2000-node complete graph of (A) our topic-weighted HK model with $\omega = 0.1$ and (B) the baseline HK model. Both simulations use $c = 0.075$ and the same set of initial opinions from a truncated Gaussian $G(0.22, 0.8)$ distribution.
    }
    \label{fig:opinion_trajectory_HK_complete_gaussian_corr0.8}
\end{figure}

\subsection{SBM graphs with Gaussian initial opinions}\label{sec:appendix:HK_sbm_gaussian}

In \cref{fig:HK_heatmap_2_community_gaussian}, we show the final order parameters $Q(T_f)$ (see \eqref{eq:order_param}) of our simulations of our topic-weighted HK model on two-community SBM graphs with truncated Gaussian $G(0.22, 0)$ initial opinions for various values of BCM parameters (namely, the confidence bound $c$ and topic weight $\omega$). 
We observe the same trends as our two-community SBM simulations with uniform initial opinions (see \cref{sec:results:hk_sbm}).
In particular, for both the community-dependent and randomly-assigned initial opinions, we tend to observe larger $Q(T_f)$ (and correspondingly less opinion fragmentation) as we either (1) increase $c$ for fixed $\omega$; or (2) decrease $\omega$ for fixed $c$.
Additionally, for fixed $\omega$, the minimum value of $c$ for our simulations to always reach consensus is larger for the community-dependent initial opinions (see \cref{fig:HK_heatmap_2_community_gaussian}A) compared to the randomly-assigned initial opinions \cref{fig:HK_heatmap_2_community_gaussian}B).

\begin{figure}[htbp!]
  \centering
  \includegraphics[width=0.8\textwidth]{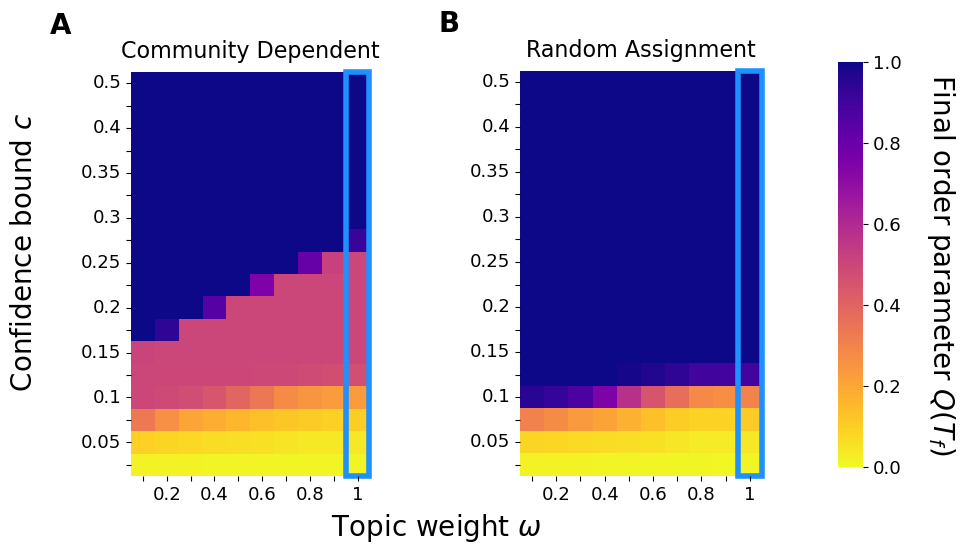}
  \caption{Final order parameters $Q(T_f)$ (see \eqref{eq:order_param}) in simulations of our topic-weighted HK model on two-community SBM graphs with $G(0.22, 0)$ initial opinions for various values of the BCM parameters $c$ and $\omega$. We simulate our model with (A) community-dependent assignment of initial opinions and (B) random assignment of initial opinions.
  }
  \label{fig:HK_heatmap_2_community_gaussian} 
\end{figure}

\section{Additional results of our topic-weighted DW simulations}\label{sec:appendix:DW_results}

In this section, we show and discuss some plots of opinion trajectories in our topic-weighted DW model on a 500-node complete graph.

In \cref{fig:DW_opinion_trajectory_uniform}, we show some plots of the opinion trajectories for simulations with uniform initial opinions and $c = 0.1$ for both our topic-weighted DW model with $\omega = 0.1$ and the baseline DW model.
Unlike in the plots of opinion trajectories for our topic-weighted HK model (see \cref{fig:opinion_trajectory_HK_complete_uniform} for example), the trajectory of a node's opinion with time does not visually appear smooth. 
This is because in the DW update rule (see \eqref{eq: DW_update_rule}), at each time, one pair of nodes is selected for interaction on one topic.
As a result, each opinion change affects only the opinions of two interacting nodes on one topic.
Consequently, the opinion trajectories consist of a series of strictly vertical or strictly horizontal line segments that correspond to each opinion change.
For our topic-weighted HK model on a complete graph with uniform initial opinions (see \cref{fig:opinion_trajectory_HK_complete_uniform}), we observe that the opinion space can be divided up into regions such that nodes with initial opinions in the same region would be in the same final opinion cluster. We do not observe this for our topic-weighted DW model.
In \cref{fig:DW_opinion_trajectory_uniform}, there are many instances where nodes with very close initial opinions are in different final opinion clusters. Furthermore, there are many instances where the opinion trajectories of nodes in different final opinion clusters cross each other. 
The final opinion cluster to which a node belongs is influenced by the randomness in selecting a topic and a pair of adjacent nodes in each interaction.

\begin{figure}[htbp]
  \centering
  \includegraphics[width=0.75\textwidth]{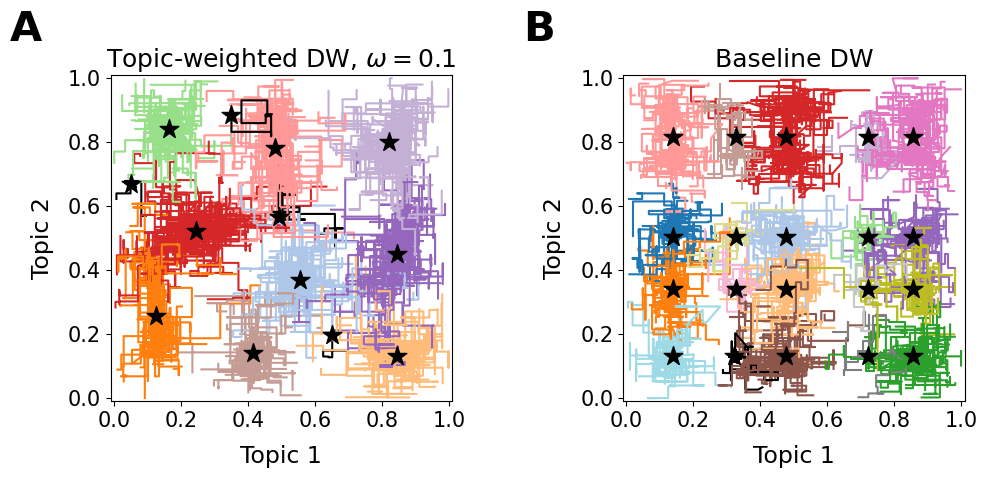}
  \caption{Opinion trajectory plots for simulations on a 500-node complete graph of (A) our topic-weighted DW model with $\omega = 0.1$ and (B) the baseline DW model. Both simulations use $c = 0.1$, the same set of initial opinions from the uniform initial opinion distribution, and the same sequence of pairs of nodes selected for interaction. 
  }
  \label{fig:DW_opinion_trajectory_uniform} 
\end{figure}

In \cref{fig:DW_opinion_trajectory_wedge}, we show some plots of the opinion trajectories for simulations of our topic-weighted DW model with $\omega = 0.4$ and $c = 0.25$ with wedge initial opinions. Each panel shows a simulation with a different set of initial opinions, and we show examples of simulations that have 1--3 final opinion clusters. 
In our topic-weighted DW model, there is randomness in which pair of nodes and which topic is selected at each time for the update rule \eqref{eq: DW_update_rule}.
Consequently, nodes in our topic-weighted DW model can have opinion trajectories that traverse more of the opinion space than nodes in our topic-weighted HK model (see \cref{fig:opinion_trajectory_HK_complete_wedge}).
Our simulations of our topic-weighted DW model with wedge initial opinions typically have 2--5 final opinion clusters.
Randomness in the selection of nodes and topics contributes to how many final opinion clusters are in a simulation. 
As we discussed in \cref{sec:results_dw}, for $c = 0.25$, there appears to be noticeably less opinion fragmentation for $0.4 \leq \omega \leq 0.6$ than for other values of $\omega$. It may be that for $c = 0.25$, for the wedge initial opinion distribution, the shape and size of the regions of receptiveness when $0.4 \leq \omega \leq 0.6$ make it more likely to have fewer final opinion clusters and smaller final order parameter than other values of $\omega$.

\begin{figure}[htbp]
  \centering
  \includegraphics[width=0.98\textwidth]{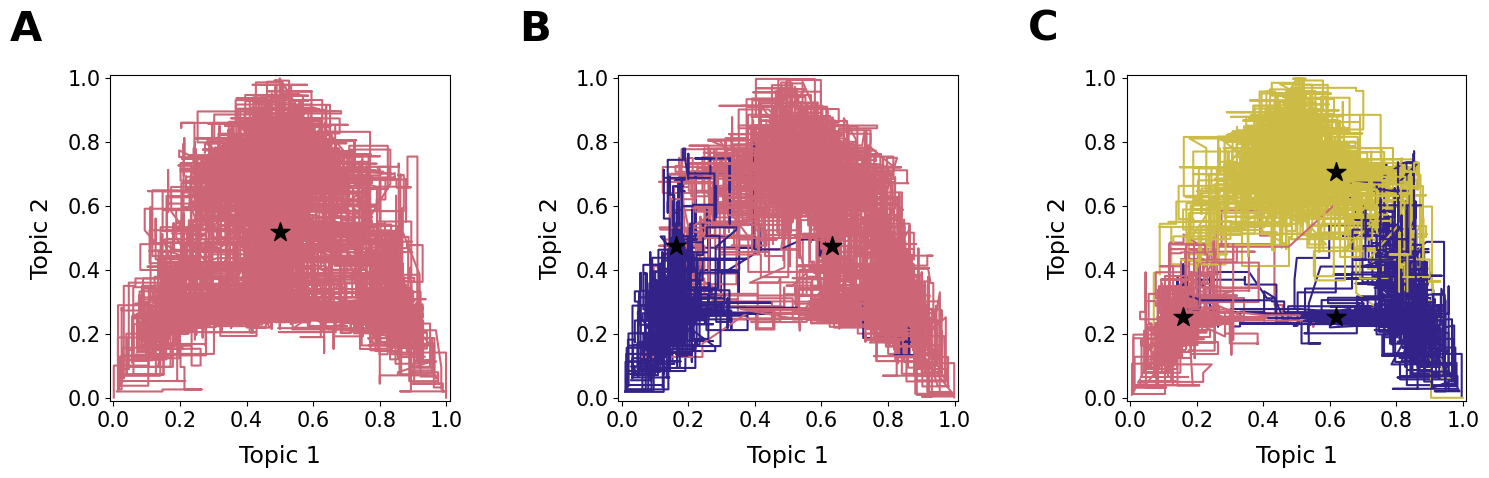}
  \caption{Opinion trajectory plots for simulations of our topic-weighted DW model on a 500-node complete graph with $\omega = 0.4$ and $c = 0.25$. Each panel shows a simulation with a different set of initial opinions from the wedge initial opinion distribution. 
  }
  \label{fig:DW_opinion_trajectory_wedge} 
\end{figure}

\section*{Acknowledgments}
We thank Mason A. Porter and Dominic Yang for helpful discussions. We also thank Deanna Needell and Jacob Foster for helpful comments. Additionally, we thank our anonymous referee for their helpful comments.

\bibliographystyle{siamplain}
\bibliography{references}

\begin{thebibliography}{10}

\bibitem{adams2022}
{\sc J.~A. Adams, G.~White, and R.~P. Araujo}, {\em Mathematical measures of
  societal polarisation}, PLoS ONE, 17 (2022), e0275283.

\bibitem{Axelrod1997}
{\sc R.~Axelrod}, {\em The dissemination of culture: {A} model with local
  convergence and global polarization}, Journal of Conflict Resolution, 41
  (1997), pp.~203--226.

\bibitem{bartlett1951}
{\sc M.~S. Bartlett}, {\em An inverse matrix adjustment arising in discriminant
  analysis}, The Annals of Mathematical Statistics, 22 (1951), pp.~107--111.

\bibitem{baumann2021}
{\sc F.~Baumann, P.~Lorenz-Spreen, I.~M. Sokolov, and M.~Starnini}, {\em
  Emergence of polarized ideological opinions in multidimensional topic
  spaces}, Physical Review X, 11 (2021), 011012.

\bibitem{ben2003bifurcations}
{\sc E.~Ben-Naim, P.~L. Krapivsky, and S.~Redner}, {\em Bifurcations and
  patterns in compromise processes}, Physica D, 183 (2003), pp.~190--204.

\bibitem{benoit2006_book}
{\sc K.~Benoit and M.~Laver}, {\em Party Policy in Modern Democracies},
  Routledge, London, United Kingdom, 2006.

\bibitem{benoit2012}
{\sc K.~Benoit and M.~Laver}, {\em The dimensionality of political space:
  {E}pistemological and methodological considerations}, European Union
  Politics, 13 (2012), pp.~194--218.

\bibitem{bernardo2024}
{\sc C.~Bernardo, C.~Altafini, A.~Proskurnikov, and F.~Vasca}, {\em Bounded
  confidence opinion dynamics: {A} survey}, Automatica, 159 (2024), 111302.

\bibitem{bhattacharyya2013}
{\sc A.~Bhattacharyya, M.~Braverman, B.~Chazelle, and H.~L. Nguyen}, {\em On
  the convergence of the {H}egselmann--{K}rause system}, in Proceedings of the
  4th Conference on Innovations in Theoretical Computer Science, ITCS '13,
  Berkeley, CA, USA, 2013, Association for Computing Machinery, New York, NY,
  USA, pp.~61--66.

\bibitem{blondel2010}
{\sc V.~D. Blondel, J.~M. Hendrickx, and J.~N. Tsitsiklis}, {\em
  Continuous-time average-preserving opinion dynamics with opinion-dependent
  communications}, SIAM Journal on Control and Optimization, 48 (2010),
  pp.~5214--5240.

\bibitem{bramson2016}
{\sc A.~Bramson, P.~Grim, D.~J. Singer, S.~Fisher, W.~Berger, G.~Sack, and
  C.~Flocken}, {\em Disambiguation of social polarization concepts and
  measures}, The Journal of Mathematical Sociology, 40 (2016), pp.~80--111.

\bibitem{brooks2020}
{\sc H.~Z. Brooks and M.~A. Porter}, {\em A model for the influence of media on
  the ideology of content in online social networks}, Physical Review Research,
  2 (2020), 023041.

\bibitem{carro2013}
{\sc A.~Carro, R.~Toral, and M.~San~Miguel}, {\em The role of noise and initial
  conditions in the asymptotic solution of a bounded confidence,
  continuous-opinion model}, Journal of Statistical Physics, 151 (2013),
  pp.~131--149.

\bibitem{selective_exposure_def}
{\sc D.~Chandler and R.~Munday}, {\em A Dictionary of Media and Communication},
  Oxford University Press, Oxford, UK, 2011.

\bibitem{chazelle2017_inertial_hk}
{\sc B.~Chazelle and C.~Wang}, {\em Inertial hegselmann-krause systems}, IEEE
  Transactions on Automatic Control, 62 (2017), pp.~3905--3913.

\bibitem{chen2021_multi_agent}
{\sc T.~Chen, Y.~Wang, J.~Yang, and G.~Cong}, {\em Modeling multidimensional
  public opinion polarization process under the context of derived topics},
  International Journal of Environmental Research and Public Health, 18 (2021),
  472.

\bibitem{musco2021}
{\sc J.~U. Christopher~Musco, Indu~Ramesh and R.~T. Witter}, {\em How to
  quantify polarization in models of opinion dynamics}, in Proceedings of the
  17th International Workshop on Mining and Learning with Graphs (MLG),
  Washington DC, USA, 2022, Association for Computing Machinery, New York, NY,
  USA.

\bibitem{chu2023density}
{\sc W.~Chu and M.~A. Porter}, {\em A density description of a
  bounded-confidence model of opinion dynamics on hypergraphs}, SIAM Journal on
  Applied Mathematics, 83 (2023), pp.~2310--2328.

\bibitem{converse1964}
{\sc P.~E. Converse}, {\em The nature of belief systems in mass publics
  (1964)}, Critical Review, 18 (2006), pp.~1--74.

\bibitem{deffuant2000mixing}
{\sc G.~Deffuant, D.~Neau, F.~Amblard, and G.~Weisbuch}, {\em Mixing beliefs
  among interacting agents}, Advances in Complex Systems, 3 (2000), pp.~87--98.

\bibitem{etesami2013_hk_convergence}
{\sc S.~R. Etesami, T.~Başar, A.~Nedić, and B.~Touri}, {\em Termination time
  of multidimensional {H}egselmann--{K}rause opinion dynamics}, in 2013
  American Control Conference, Washington, DC, USA, 2013, Institute of
  Electrical and Electronics Engineers (IEEE), Piscataway, NJ, USA,
  pp.~1255--1260.

\bibitem{fortunato2005}
{\sc S.~Fortunato}, {\em On the consensus threshold for the opinion dyanamics
  of {K}rause--{H}egselmann}, International Journal of Modern Physics C, 16
  (2005), pp.~259--270.

\bibitem{fortunato2005_vector}
{\sc S.~Fortunato, V.~Latora, A.~Pluchino, and A.~Rapisarda}, {\em Vector
  opinion dynamics in a bounded confidence consensus model}, International
  Journal of Modern Physics C, 16 (2005), pp.~1535--1551.

\bibitem{friedkin1990social}
{\sc N.~E. Friedkin and E.~C. Johnsen}, {\em Social influence and opinions},
  Journal of Mathematical Sociology, 15 (1990), pp.~193--206.

\bibitem{gilbert1959_gnp}
{\sc E.~N. Gilbert}, {\em Random graphs}, The Annals of Mathematical
  Statistics, 30 (1959), pp.~1141--1144.

\bibitem{gomez-serrano2012}
{\sc J.~G\'omez-Serrano, C.~Graham, and J.-Y. {Le Boudec}}, {\em The bounded
  confidence model of opinion dynamics}, Mathematical Models and Methods in
  Applied Sciences, 22 (2012), 1150007.

\bibitem{gubanov_multidimensional_2021}
{\sc D.~A. Gubanov, I.~V. Petrov, and A.~G. Chkhartishvili}, {\em
  Multidimensional model of opinion dynamics in social networks: {P}olarization
  indices}, Automation and Remote Control, 82 (2021), pp.~1802--1811.

\bibitem{hagner1989}
{\sc W.~W. Hager}, {\em Updating the inverse of a matrix}, SIAM Review, 31
  (1989), pp.~221--239.

\bibitem{HK_model}
{\sc R.~Hegselmann and U.~Krause}, {\em Opinion dynamics and bounded confidence
  models, analysis and simulation}, Journal of Artificial Societies and Social
  Simulation, 5(3) (2002), 2.

\bibitem{hegselmann2019_hk_d_dimension_convergence}
{\sc R.~Hegselmann and U.~Krause}, {\em Consensus and fragmentation of opinions
  with a focus on bounded confidence}, The American Mathematical Monthly, 126
  (2019), pp.~700--716.

\bibitem{hickok2022bounded}
{\sc A.~Hickok, Y.~Kureh, H.~Z. Brooks, M.~Feng, and M.~A. Porter}, {\em A
  bounded-confidence model of opinion dynamics on hypergraphs}, SIAM Journal on
  Applied Dynamical Systems, 21 (2022), pp.~1--32.

\bibitem{holland1983_sbm}
{\sc P.~W. Holland, K.~B. Laskey, and S.~Leinhardt}, {\em Stochastic
  blockmodels: {F}irst steps}, Social Networks, 5 (1983), pp.~109--137.

\bibitem{jacobmeier2006}
{\sc D.~Jacobmeier}, {\em Focusing of opinions in the {D}effaunt model: {F}irst
  impression counts}, International Journal of Modern Physics C, 17 (2006),
  pp.~1801--1808.

\bibitem{jager2005}
{\sc W.~Jager and F.~Amblard}, {\em Uniformity, bipolarization and pluriformity
  captured as generic stylized behavior with an agent-based simulation model of
  attitude change}, Computational {\&} Mathematical Organization Theory, 10
  (2005), pp.~295--303.

\bibitem{jia2015opinion}
{\sc P.~Jia, A.~MirTabatabaei, N.~E. Friedkin, and F.~Bullo}, {\em Opinion
  dynamics and the evolution of social power in influence networks}, SIAM
  Review, 57 (2015), pp.~367--397.

\bibitem{kou2012_hk}
{\sc G.~Kou, Y.~Zhao, Y.~Peng, and Y.~Shi}, {\em Multi-level opinion dynamics
  under bounded confidence}, PLOS ONE, 7 (2012), e43507.

\bibitem{krause2000}
{\sc U.~Krause}, {\em A discrete nonlinear and non-autonomous model of
  consensus}, in Communications in Difference Equations: Proceedings of the
  Fourth International Conference on Difference Equations, S.~N. Elaydi,
  J.~Popenda, and J.~Rakowski, eds., CRC Press, Amsterdam, The Netherlands,
  2000, pp.~227--236.

\bibitem{li2017_ja_vector}
{\sc J.~Li and R.~Xiao}, {\em Agent-based modelling approach for
  multidimensional opinion polarization in collective behaviour}, Journal of
  Artificial Societies and Social Simulation, 20 (2017), 4.

\bibitem{lorenz2005stabilization}
{\sc J.~Lorenz}, {\em A stabilization theorem for dynamics of continuous
  opinions}, Physica A: Statistical Mechanics and its Applications, 355 (2005),
  pp.~217--223.

\bibitem{lorenz2006_multi}
{\sc J.~Lorenz}, {\em Continuous opinion dynamics of multidimensional
  allocation problems under bounded confidence. {M}ore dimensions lead to
  better chances for consensus}, European Journal of Economic and Social
  Systems, 19 (2006), pp.~213--227.

\bibitem{lorenz2007continuous}
{\sc J.~Lorenz}, {\em Continuous opinion dynamics under bounded confidence: {A}
  survey}, International Journal of Modern Physics C, 18 (2007),
  pp.~1819--1838.

\bibitem{lorenz_fostering_2008}
{\sc J.~Lorenz}, {\em Fostering consensus in multidimensional continuous
  opinion dynamics under bounded confidence}, in Managing {Complexity}:
  {Insights}, {Concepts}, {Applications}, D.~Helbing, ed., Springer, Berlin,
  Heidelberg, 2008, pp.~321--334.

\bibitem{meng2018}
{\sc X.~F. Meng, R.~A. Van~Gorder, and M.~A. Porter}, {\em Opinion formation
  and distribution in a bounded-confidence model on various networks}, Physical
  Review E, 97 (2018), 022312.

\bibitem{newman2018_book}
{\sc M.~Newman}, {\em Networks}, Oxford University Press, Oxford, United
  Kingdom, 2nd~ed., 2018.

\bibitem{noorazar2020_review}
{\sc H.~Noorazar, K.~R. Vixie, A.~Talebanpour, and Y.~Hu}, {\em From classical
  to modern opinion dynamics}, International Journal of Modern Physics C, 31
  (2020), 2050101.

\bibitem{parasnis_hegselmann-krause_2018}
{\sc R.~Parasnis, M.~Franceschetti, and B.~Touri}, {\em Hegselmann--{Krause}
  dynamics with limited connectivity}, in 2018 {IEEE} Conference on Decision
  and Control ({CDC}), 2018, pp.~5364--5369.

\bibitem{parsegov2015}
{\sc S.~E. Parsegov, A.~V. Proskurnikov, R.~Tempo, and N.~E. Friedkin}, {\em A
  new model of opinion dynamics for social actors with multiple interdependent
  attitudes and prejudices}, in 2015 54th {IEEE} Conference on Decision and
  Control ({CDC}), Osaka, Japan, 2015, Institute of Electrical and Electronics
  Engineers (IEEE), Piscataway, NJ, USA, pp.~3475--3480.

\bibitem{porter2016}
{\sc M.~A. Porter and J.~P. Gleeson}, {\em Dynamical Systems on Networks: {A}
  {T}utorial}, vol.~4 of Frontiers in Applied Dynamical Systems: Reviews and
  Tutorials, Springer International Publishing, Cham, Switzerland, 2016.

\bibitem{schawe2021_HK_networks}
{\sc H.~Schawe, S.~Fontaine, and L.~Hern\'andez}, {\em When network bridges
  foster consensus. {B}ounded confidence models in networked societies},
  Physical Review Research, 3 (2021), 023208.

\bibitem{schweighofer2020}
{\sc S.~Schweighofer, D.~Garcia, and F.~Schweitzer}, {\em An agent-based model
  of multi-dimensional opinion dynamics and opinion alignment}, Chaos: {A}n
  Interdisciplinary Journal of Nonlinear Science, 30 (2020), 093139.

\bibitem{shang2013_dw_critical}
{\sc Y.~Shang}, {\em {D}effuant model with general opinion distributions:
  {F}irst impression and critical confidence bound}, Complexity, 19 (2013),
  pp.~38--49.

\bibitem{sobkowicz2015}
{\sc P.~Sobkowicz}, {\em Extremism without extremists: {D}effuant model with
  emotions}, Frontiers in Physics, 3 (2015), 17.

\bibitem{Tan2018}
{\sc Z.~X. Tan and K.~H. Cheong}, {\em Cross-issue solidarity and truth
  convergence in opinion dynamics}, Journal of Physics A: Mathematical and
  Theoretical, 51 (2018).

\bibitem{urena2019review}
{\sc R.~Urena, G.~Kou, Y.~Dong, F.~Chiclana, and E.~Herrera-Viedma}, {\em A
  review on trust propagation and opinion dynamics in social networks and group
  decision making frameworks}, Information Sciences, 478 (2019), pp.~461--475.

\bibitem{wang2017_order_param}
{\sc C.~Wang, Q.~Li, W.~E, and B.~Chazelle}, {\em Noisy {H}egselmann--{K}rause
  systems: {P}hase transition and the {2R}-conjecture}, Journal of Statistical
  Physics, 166 (2017), pp.~1209--1225.

\bibitem{yang2014_hk_consensus}
{\sc Y.~Yang, D.~V. Dimarogonas, and X.~Hu}, {\em Opinion consensus of modified
  {H}egselmann--{K}rause models}, Automatica, 50 (2014), pp.~622--627.

\bibitem{ye2020}
{\sc M.~Ye, M.~H. Trinh, Y.-H. Lim, B.~D. Anderson, and H.-S. Ahn}, {\em
  Continuous-time opinion dynamics on multiple interdependent topics},
  Automatica, 115 (2020), 108884.

\end{thebibliography}

\end{document}